\documentclass
[11pt,letterpaper]
{article}
\newif\iffull

\fulltrue

\usepackage[notes=true,later=false,camera=true]{dtrt}
   \usepackage{hyperref}
   \hypersetup{hyperindex=true,pdfpagemode=UseOutlines,bookmarksnumbered=true,bookmarksopen=true,bookmarksopenlevel=2,pdfstartview=FitH,pdfborder={0 0 1},linkbordercolor=dt@linkcolor,citebordercolor=dt@linkcolor,urlbordercolor=dt@linkcolor}
\usepackage[utf8]{inputenc}
\usepackage{etex}
\usepackage{xspace,enumerate}
\usepackage[T1]{fontenc}
\usepackage[full]{textcomp}
\usepackage[american]{babel}
\usepackage{mathtools}
\usepackage{titling}
\usepackage{amsthm}

\usepackage[
letterpaper,
top=1.2in,
bottom=1.2in,
left=1in,
right=1in]{geometry}
\usepackage{newpxtext} 
\usepackage{textcomp} 
\usepackage[varg,bigdelims]{newpxmath}
\usepackage[scr=rsfso]{mathalfa}
\usepackage{bm} 
\let\mathbb\varmathbb
\usepackage{microtype}

\usepackage[capitalise,nameinlink]{cleveref}
\crefname{lemma}{Lemma}{Lemmas}
\crefname{fact}{Fact}{Facts}
\crefname{theorem}{Theorem}{Theorems}
\crefname{mtheorem}{Theorem}{Theorems}
\crefname{corollary}{Corollary}{Corollaries}
\crefname{claim}{Claim}{Claims}
\crefname{example}{Example}{Examples}
\crefname{algorithm}{Algorithm}{Algorithms}
\crefname{subroutine}{Subroutine}{Subroutines}
\crefname{problem}{Problem}{Problems}
\crefname{definition}{Definition}{Definitions}
\usepackage{paralist}
\usepackage{turnstile}
\usepackage{mdframed}
\usepackage{tikz}
\usepackage{caption}
\DeclareCaptionType{Algorithm}
\usepackage{newfloat}

\usepackage{physics}
\usepackage[shortlabels]{enumitem}
\setlist[enumerate,1]{label={(\arabic*)}}

\newtheorem{theorem}{Theorem}[section]
\newtheorem{mtheorem}{Theorem}
\newtheorem*{theorem*}{Theorem}

\newtheorem{proposition}[theorem]{Proposition}
\newtheorem*{proposition*}{Proposition}
\newtheorem{lemma}[theorem]{Lemma}
\newtheorem*{lemma*}{Lemma}
\newtheorem{corollary}[theorem]{Corollary}
\newtheorem*{conjecture*}{Conjecture}
\newtheorem{fact}[theorem]{Fact}
\newtheorem*{fact*}{Fact}

\newtheorem*{hypothesis*}{Hypothesis}

\theoremstyle{definition}
\newtheorem{definition}[theorem]{Definition}
\newtheorem*{definition*}{Definition}

\newtheorem{algorithm}{Algorithm}

\theoremstyle{remark}

\newtheorem*{claim*}{Claim}
\newtheorem{remark}[theorem]{Remark}
\newtheorem*{remark*}{Remark}

\newtheorem*{observation*}{Observation}

\newcommand{\Authornotecolored}[3]{}
\newcommand{\Authorcomment}[2]{}
\newcommand{\Authorfnote}[2]{}

\definecolor{forestgreen(traditional)}{rgb}{0.0, 0.27, 0.13}

\usepackage{boxedminipage}

\newcommand{\Paren}[1]{\left(#1\right)}

\newcommand{\bigabs}[1]{\bigl\lvert#1\bigr\rvert}

\renewcommand{\norm}[1]{\lVert#1\rVert}
\newcommand{\Norm}[1]{\left\lVert#1\right\rVert}
\newcommand{\bignorm}[1]{\bigl\lVert#1\bigr\rVert}

\newcommand{\bigiprod}[1]{\big\langle#1\big\rangle}
\newcommand{\Bigiprod}[1]{\Big\langle#1\Big\rangle}
\newcommand{\Esymb}{\mathbb{E}}

\DeclareMathOperator*{\E}{\Esymb}

\newcommand\bdot\bullet

\DeclareMathOperator{\poly}{poly}

\DeclareMathOperator{\polylog}{polylog}

\newcommand{\Z}{\mathbb Z}
\newcommand{\N}{\mathbb N}
\newcommand{\R}{\mathbb{R}}

\renewcommand{\P}{\mathcal P}

\newcommand{\cN}{\mathcal N}

\renewcommand{\S}{\mathbb{S}}

\renewcommand{\leq}{\leqslant}
\renewcommand{\le}{\leqslant}
\renewcommand{\geq}{\geqslant}
\renewcommand{\ge}{\geqslant}
\let\epsilon=\varepsilon
\numberwithin{equation}{section}
\newcommand\MYcurrentlabel{xxx}
\newcommand{\MYstore}[2]{%
  \global\expandafter \def \csname MYMEMORY #1 \endcsname{#2}%
}
\newcommand{\MYload}[1]{%
  \csname MYMEMORY #1 \endcsname%
}
\newcommand{\MYnewlabel}[1]{%
  \renewcommand\MYcurrentlabel{#1}%
  \MYoldlabel{#1}%
}
\newcommand{\MYdummylabel}[1]{}
\newcommand{\torestate}[1]{%
  \let\MYoldlabel\label%
  \let\label\MYnewlabel%
  #1%
  \MYstore{\MYcurrentlabel}{#1}%
  \let\label\MYoldlabel%
}
\newcommand{\restatetheorem}[1]{%
  \let\MYoldlabel\label
  \let\label\MYdummylabel
  \begin{theorem*}[Restatement of \cref{#1}]
    \MYload{#1}
  \end{theorem*}
  \let\label\MYoldlabel
}
\newcommand{\restatelemma}[1]{%
  \let\MYoldlabel\label
  \let\label\MYdummylabel
  \begin{lemma*}[Restatement of \cref{#1}]
    \MYload{#1}
  \end{lemma*}
  \let\label\MYoldlabel
}
\newcommand{\restateprop}[1]{%
  \let\MYoldlabel\label
  \let\label\MYdummylabel
  \begin{proposition*}[Restatement of \cref{#1}]
    \MYload{#1}
  \end{proposition*}
  \let\label\MYoldlabel
}
\newcommand{\restatefact}[1]{%
  \let\MYoldlabel\label
  \let\label\MYdummylabel
  \begin{fact*}[Restatement of \prettyref{#1}]
    \MYload{#1}
  \end{fact*}
  \let\label\MYoldlabel
}
\newcommand{\restate}[1]{%
  \let\MYoldlabel\label
  \let\label\MYdummylabel
  \MYload{#1}
  \let\label\MYoldlabel
}

\newcommand{\eps}{\epsilon}

\allowdisplaybreaks
\newcommand*{\Id}{\mathbb{I}}

\newcommand{\pE}{\mathop{\tilde{\mathbb E}}}

\newcommand{\1}{\bm{1}}

\let\svthefootnote\thefootnote
\newcommand\blfootnote[1]{%
  \let\thefootnote\relax%
  \footnotetext{#1}%
  \let\thefootnote\svthefootnote%
}

\title{Polynomial-Time Sum-of-Squares Can Robustly Estimate Mean and Covariance of Gaussians Optimally}

\author{Pravesh K.\ Kothari\thanks{Supported by NSF CAREER Award \#2047933 and an award from the Google Research Scholar program.} \and Peter Manohar\thanks{Supported in part by an ARCS Fellowship, NSF Graduate Research Fellowship (under Grant No.\ DGE1745016) and NSF CCF-1814603. 
} \and Brian Hu Zhang\thanks{This material is based on work supported by the National Science Foundation under grants IIS-1718457, IIS-1901403, and CCF-1733556, and the ARO under award W911NF2010081.}}

	\date{\texttt{\{praveshk,pmanohar,bhzhang\}@cs.cmu.edu} \\ Computer Science Department \\ Carnegie Mellon University \\ Pittsburgh, USA}

\usepackage{autonum} 
\makeatletter
\autonum@generatePatchedReferenceCSL{Cref}
\makeatother
\makeatletter
\newcommand{\restore@Environment}[1]{%
  \AtBeginDocument{%
    \csletcs{#1*}{#1}%
    \csletcs{end#1*}{end#1}%
  }%
}
\forcsvlist\restore@Environment{alignat,equation,gather,multline,flalign,align}
\makeatother

\allowdisplaybreaks

\newcommand{\keywords}[1]{\bigskip\par\noindent{\footnotesize\textbf{Keywords\/}: #1}}

 \renewcommand{\abs}[1]{\lvert #1 \rvert}

 \renewcommand{\ip}[1]{\langle #1 \rangle}

\newcommand{\tvdist}{\Delta_{\text{TV}}}
\newcommand{\degfin}{12}
\newcommand{\deghalf}{6}
\newcommand{\sampmean}{\mu_0}

\begin{document}

\pagestyle{empty}


\maketitle\blfootnote{Any opinions, findings, and conclusions or recommendations expressed in this material are those of the author(s) and do not necessarily reflect the views of the National Science Foundation.}
\thispagestyle{empty} 


\begin{abstract}
In this work, we revisit the problem of estimating the mean and covariance of an unknown $d$-dimensional Gaussian distribution in the presence of an $\epsilon$-fraction of adversarial outliers. The work of \cite{DiakonikolasKK016} gave a polynomial time algorithm for this task with optimal $\tilde{O}(\eps)$ error using $n = \poly(d, 1/\eps)$ samples. 

On the other hand, \cite{KS17} introduced a general framework for robust moment estimation via a canonical sum-of-squares relaxation that succeeds for the more general class of \emph{certifiably subgaussian} and \emph{certifiably hypercontractive}~\cite{BakshiK20} distributions. When specialized to Gaussians, this algorithm obtains the same $\tilde{O}(\epsilon)$ error guarantee as \cite{DiakonikolasKK016} but incurs a super-polynomial sample complexity ($n = d^{O(\log 1/\eps)}$) and running time ($n^{O(\log(1/\eps))}$). This cost appears inherent to their analysis as it relies only on sum-of-squares certificates of upper bounds on directional moments while the analysis in \cite{DiakonikolasKK016} relies on \emph{lower bounds} on directional moments inferred from algebraic relationships between moments of Gaussian distributions. 

We give a new, simple analysis of the \emph{same} canonical sum-of-squares relaxation used in \cite{KS17,BakshiK20} and show that for Gaussian distributions, their algorithm achieves the same error, sample complexity and running time guarantees as of the specialized algorithm in~\cite{DiakonikolasKK016}. Our key innovation is a new argument that allows using moment lower bounds without having sum-of-squares certificates for them. We believe that our proof technique will likely be useful in designing new robust estimation algorithms. 

\end{abstract}
\clearpage


  \microtypesetup{protrusion=false}
  \tableofcontents{}
  \microtypesetup{protrusion=true}

\clearpage

\pagestyle{plain}
\setcounter{page}{1}
\section{Introduction}
Designing estimation algorithms for estimating basic parameters of probability distributions from samples is a foundational computational problem in machine learning. However, natural estimation algorithms, such as taking the sample mean for population mean, can be brittle -- even a single outlier in the data can lead to an arbitrarily large estimation error. In the 1960s, Tukey and Huber began systematic efforts to build \emph{robust estimators} that can tolerate minor deviations of the input from the chosen model, such as the injection of a small constant fraction of adversarially chosen outliers into the sample. While this effort has led to a burgeoning body of work called \emph{robust statistics}, the algorithms from this line of work typically require exponential time in the underlying dimension to succeed and are thus inefficient in high-dimensional settings. 

In 2016, two papers~\cite{DiakonikolasKK016,LaiRV16} pioneered a systematic effort to build \emph{computationally efficient} robust estimators. Since their work, the study of \emph{algorithmic robust statistics} has evolved into an active area, that, in addition to yielding concrete solutions to basic robust estimation problems, has led to the synthesis of truly new algorithmic ideas (often improving even the classical, non-robust algorithms) that identify and clarify general principles for efficient robust estimation. 

A key insight from this line of work has been a general blueprint for robust estimation using the \emph{sum-of-squares} (SoS) method. A sequence of works in 2018 gave a canonical sum-of-squares relaxation and a rounding algorithm that gives the nearly statistically optimal outlier-robust estimation of \emph{moments}~\cite{KS17} and robust clustering~\cite{KothariSteinhardt17,DBLP:conf/stoc/Hopkins018} of spherical mixtures of a broad class of probability distributions. Since then, this framework has been refined and expanded to obtain state-of-the-art robust estimation algorithms for problems such as outlier-robust regression~\cite{KlivansKM18,BakshiP21}, clustering of non-spherical mixtures~\cite{BakshiK20,DBLP:conf/focs/BakshiDHKKK20}, heavy tailed estimation~\cite{Hopkins,DBLP:conf/stoc/CherapanamjeriH20}, list-decodable regression and subspace recovery~\cite{DBLP:conf/nips/KarmalkarKK19,BakshiK20,DBLP:conf/colt/RaghavendraY20,DBLP:conf/soda/RaghavendraY20} and robust learning of a mixture of arbitrary Gaussians~\cite{DBLP:conf/stoc/LiuM21,DBLP:journals/corr/abs-2012-02119}. 

In addition, algorithms from the SoS-based robust estimation framework have the advantage of abstracting out natural analytic properties of the statistical model in question and yielding robust estimators for all distributions that satisfy such properties in a blackbox way. For example, the algorithms for robust estimation of moments~\cite{KS17} and clustering~\cite{DBLP:conf/stoc/Hopkins018,KothariSteinhardt17} apply to all \emph{certifiably subgaussian} distributions, that, informally speaking, are distributions that admit ``sum-of-squares certificates'' of the property of having subgaussian low order moments. As another example, the covariance estimation algorithm of~\cite{BakshiK20} applies to all distributions that admit sum-of-squares certificates of bounds on moments of degree-$2$ polynomials (\emph{certifiable hypercontractivity}). Such properties are already known to hold for a broad class of distributions and verifying them for a new family immediately generalizes such results. In fact, one can interpret the analysis in the sum-of-squares framework as identifying structural properties (certificates of appropriate analytic properties) of the  distribution families that can be exploited for the design of efficient robust estimation algorithms.\footnote{See this \href{https://simons.berkeley.edu/talks/recent-progress-algorithmic-robust-statistics-sum-squares-method}{recent talk} for this perspective and its applications to weakening distributional assumptions in robust estimation.}

\paragraph{Robust Mean Estimation for Gaussians.} While the SoS-based framework above typically achieves the best known recovery guarantees (among polynomial time algorithms), a striking exception so far has been the task of robustly estimating the mean and covariance of an unknown Gaussian distribution. In this problem, the algorithm is given input data $Y = \{y_1, y_2,\ldots, y_n\} \subseteq \R^d$ that is obtained by \emph{arbitrarily and adversarially} corrupting $\epsilon n$ points in an i.i.d.\ sample $X = \{x_1, x_2, \ldots, x_n\}$ from an unknown Gaussian distribution $\cN(\mu,\Sigma)$. The algorithm of~\cite{DiakonikolasKK016} obtains estimates $\hat{\mu},\hat{\Sigma}$ so that the total variation distance $d_{TV}(\cN(\mu, \Sigma), \cN(\hat{\mu},\hat{\Sigma})) \leq \tilde{O}(\epsilon)$. This is optimal up to logarithmic factors in $\epsilon$ in the bound (and there is evidence~\cite{DiakonikolasKS17} that such a loss might be necessary for polynomial time algorithms). Their algorithm requires $n=\poly(d,1/\epsilon)$ samples and $\poly(n)$ running time. On the other hand, the best known SoS-based algorithm for the problem is obtained by specializing the analysis in~\cite{BakshiK20} for mean and covariance estimation of \emph{certifiably hypercontractive} distributions to the case of Gaussians. This analysis yields the same error bound of $\tilde{O}(\epsilon)$ on the total variation error but requires super-polynomially many $d^{O(\log 1/\epsilon)}$ samples and $n^{O(\log 1/\epsilon)}$ running time. 

There is an important technical bottleneck in the analysis of the canonical SoS algorithm for obtaining the stronger guarantees in~\cite{DiakonikolasKK016}. The analysis in~\cite{KS17} (and extensions in~\cite{BakshiK20}) only uses upper bounds on the higher moments of distributions. On the other hand, the stronger analysis in~\cite{DiakonikolasKK016} implicitly relies on a non-trivial \emph{lower bound} on moments of arbitrary subsets of the original sample of size $(1-\epsilon)n$. The best known sum-of-squares certificates for such a lower bound property appear to require an exponential cost in $O(\log 1/\epsilon)$ in both sample complexity and running time.
And, it is plausible (though, still unproven) that such a cost is necessary! This state of affairs leads us to the main motivating question of this work: 

\emph{Can the canonical SoS based algorithm give a robust estimate with $\tilde{O}(\epsilon)$ total variation error for the mean and covariance of Gaussian distributions in polynomial time and samples? Or is the SoS framework for moment estimation weaker, when specialized to Gaussian distributions?} 

In this work, we give a new analysis of the canonical sum-of-squares-based algorithm for robust mean estimation for Gaussians (that only has subgaussian upper bounds on $4$th moment as constraints) that recovers the polynomial running time and sample complexity guarantees of~\cite{DiakonikolasKK016} and same error up to $\poly \log 1/\epsilon$ factors. Our key innovation (that we explain later in this section) is a new argument that works around the issue of finding efficient sum-of-squares certificates for moment lower bounds and yet manages to prove the stronger guarantee. We believe that this new technique will likely find further applications in efficient robust estimation. 

\subsection{Our results}
Formally, our algorithms work in the following strong contamination model for corrupted samples used in several prior works on robust estimation, beginning with~\cite{DiakonikolasKK016,LaiRV16}. 

\begin{definition}[Strong contamination model]
Let $D$ be a distribution on $\R^d$ and let $X = \{x_1, x_2, \ldots,x_n\}$ be an i.i.d.\ sample from $D$. 
In the strong contamination model, an $\epsilon$-corrupted sample is obtained by replacing any adversarially chosen $\epsilon n$ points from $X$ with arbitrary outliers to obtain $Y = \{y_1, y_2,\ldots,y_n\}$. 
\end{definition}

Our main result is an analysis of the following canonical SoS relaxation for mean and covariance estimation along with a  simple rounding (used in~\cite{KS17,BakshiK20}) for estimating the mean and covariance of an unknown distribution.

\begin{algorithm}[Mean and spectral norm covariance estimation]
  \begin{mdframed}
  \label{program:canonical}
      \mbox{}
    \begin{description}
    \item[Input:] Parameter $\eps \in (0,1)$, and corrupted samples $y_1, \dots, y_{n} \in \R^d$.
    \item[Operation:] 
         Find a degree-$\degfin$ pseudo-expectation $\pE$ (solution to the SoS semidefinite programming relaxation) in  variables $x'_1, \dots, x'_n \in \R^d$, ~$w_1, \dots, w_n \in \R$, ~$\mu' := \E_i x'_i$, ~$\Sigma' := \E_{i} (x'_i - \mu')(x'_i - \mu')^{\top}$ satisfying the following set of constraints:
            \begin{enumerate}[(1)]
        \item \textbf{Booleanity of intersection Variables:} $w_i^2 = w_i$ for every $i \in [n]$,
        \item \textbf{Size of intersection:} $\sum_{i = 1}^n w_i = (1 - \eps)n$,
        \item \textbf{Intersection constraints:} $w_i x'_i = w_i y_i$ for every $i \in [n]$,
        \item \textbf{Certifiable Subgaussianity of 4th moments:} \\ $\frac{1}{n} \sum_{i =1}^n \Paren{\ip{x'_i - \mu',v}^2 - v^{\top}\Sigma'v}^2 \leq (2 + \tilde{O}(\eps)) (v^{\top} \Sigma' v)^2$ for every $v \in \R^d$.
      \end{enumerate}
  \item[Output:] $\hat{\mu} = \pE[\mu'], \hat{\Sigma} = \pE[\Sigma']$.
    \end{description}
         \end{mdframed}
  \end{algorithm}

The constraints of the program encode the task of finding a set of points $X'$ that intersects the input sample $Y$ in $(1-\epsilon)n$ points (encoded by the first $3$ sets of constraints) such that the empirical 4-th moments of $X'$ are bounded above by at most $\sim 2$ times the square of the 2nd moments (the last set of constraints). The last set of constraints, though apparently infinitely many (one for every $v \in \R^d$) admit a succinct representation via techniques of~\cite{KS17,DBLP:conf/stoc/Hopkins018} (see \cref{sec:elimination}, or~\cite{FlemingKP19}, Chapter 4 for an exposition). The intended solution for this polynomial system is $X' = X$ -- the unknown, true i.i.d.\ sample. (And then setting $w_i = \1(x_i = y_i)$, and $\mu', \Sigma'$ to be the empirical mean/covariance of $X$.) It is easy to check that $X$ satisfies the last set of constraints -- the only property of i.i.d.\ Gaussian samples that we enforce. 

We prove the following formal guarantees on \cref{program:canonical}.
\begin{mtheorem}[Mean and spectral norm covariance estimation]
\label{thm:main}
\MYstore{thm:main}{
\cref{program:canonical} takes input an $\epsilon$-corrupted sample of size $n$ from a Gaussian distribution with mean $\mu$ and covariance $\Sigma$ and in $\poly(n)$-time, outputs estimates $\hat{\mu} \in \R^d, \hat{\Sigma} \in \R^{d \times d}$ with the following guarantee. If $\Sigma \succeq 2^{-\poly(d)}I$, and $n \geq \tilde{O}(d^2 \log^5(1/\delta)/\eps^2)$, with probability at least $1 - \delta$ over the draw of the original uncorrupted sample $X$, the estimates $\hat{\mu}, \hat{\Sigma}$ satisfy:
\begin{enumerate}[(1)]
\item (Mean estimation) $\Norm{\Sigma^{-1/2}(\hat{\mu} - \mu)}_2 \leq \tilde{O}(\eps)$, and
\item (Covariance estimation in spectral norm) $(1-\tilde{O}(\epsilon))\Sigma \preceq \hat{\Sigma} \preceq (1+\tilde{O}(\epsilon)) \Sigma$.
\end{enumerate}
}
\MYload{thm:main}
\end{mtheorem}
\begin{remark}[Computational Model and Numerical Issues]
Our algorithm succeeds in the standard word RAM model of computation. In this model, the input sample $Y$ is given to the algorithm after ``truncating'' the real numbers to rational numbers with $\poly(d)$ bits of precision. The running time of our algorithm is polynomial in the size of the bit representation of the input. The assumption on the smallest eigenvalue of $\Sigma$ in the statement above is entirely an artefact of the numerical issues as the truncation of $Y$ to rational numbers, in general, does not allow recovering eigenvalues of $\Sigma$ that are not representable in polynomially many bits of precision. Such an assumption is required (but sometimes not stated explicitly) by all prior works on robust estimation when implemented in the standard word RAM model of computation. 

We note that it is possible (though, requires additional steps in the algorithm) to remove the assumption on the smallest eigenvalue of $\Sigma$ if we instead assume that the unknown $\Sigma$ has rational entries. Such an assumption is clearly necessary as algorithms in the word RAM model can only output $\hat{\Sigma}$ with rational entries. We omit the description of such a method and instead choose to make an assumption that the smallest eigenvalue of $\Sigma$ can be written down in $\poly(d)$ bits. 
\end{remark}

\cref{thm:main} shows that the algorithm of \cite{KS17}, when analyzed for Gaussian distributions, achieves the information-theoretically optimal $\tilde{O}(\eps)$ error guarantee using $n = \poly(d, 1/\eps)$ samples and $\poly(n)$ running time. This shows that the analysis of \cite{KS17}, which is tight for the more general class of certifiable subgaussian and certifiable hypercontractive distributions, can be improved in the specific case of Gaussians.

The guarantees achieved by \cref{thm:main} are weaker than the guarantees of the algorithm in \cite{DiakonikolasKK016}, whose estimate $\hat{\Sigma}$ is additionally close to $\Sigma$ in relative Frobenius error. We show that by analyzing the degree-$\degfin$ SoS relaxation of the following program (that replaces the certifiable subgaussianity constraints by certifiable hypercontractivity constraints on degree-$2$ polynomials), we can upgrade the guarantees of \cref{thm:main} to achieve the stronger Frobenius norm guarantee of \cite{DiakonikolasKK016}. We note that this program (with additional higher-degree certifiable hypercontractivity constraints) was analyzed in~\cite{BakshiK20} to obtain similar guarantees on the mean and covariance estimation of the more general class of all certifiably hypercontractive distributions, but needed $n=d^{O(\log 1/\epsilon)}$ samples and $n^{O(\log 1/\epsilon)}$ running time. Our contribution is obtaining a sharper analysis of the same program for the case of Gaussian distributions. 
  \begin{algorithm}[Frobenius norm covariance estimation]
  \label{program:frobnorm}
     \begin{mdframed}
    \mbox{}
    \begin{description}
    \item[Input:]
      Parameter $\eps \in (0,1)$, and corrupted samples $y_1, \dots, y_{n} \in \R^d$.
    \item[Operation:]
    Find a degree-$\degfin$ pseudo-expectation $\pE$ in the variables $x'_1, \dots, x'_n \in \R^d$, ~$w_1, \dots, w_n \in \R$, $\mu' := \E_i x'_i$, $\Sigma' := \E_{i} (x'_i - \mu')(x'_i - \mu')^{\top}$ satisfying the following set of constraints:
      \begin{enumerate}[(1)]
      \item $w_i^2 = w_i$ for every $i \in [n]$,
      \item $\sum_{i = 1}^n w_i = (1 - \eps)n$,
      \item $w_i x'_i = w_i y_i$ for every $i \in [n]$,
      \item $\E_i \ip{(x'_i - \mu')(x'_i - \mu')^{\top} - \Sigma', P}^2 \leq (2 + \tilde{O}(\eps)) \norm{P}_F^2$ for every symmetric $P \in \R^{d \times d}$.
      \end{enumerate}
    \item[Output:] $\hat \Sigma := \pE[\Sigma']$.
    \end{description}
      \end{mdframed}
  \end{algorithm}
  
\begin{mtheorem}[Frobenius norm covariance estimation with $\Sigma \approx \Id$]
\label{thm:frobnorm}
\MYstore{thm:frobnorm}{
\cref{program:frobnorm} takes input an $\epsilon$-corrupted sample of size $n$ from a Gaussian distribution with mean $\mu$ and covariance $\Sigma$ with $(1 - \tilde{O}(\eps)) \Id \preceq \Sigma \preceq(1 + \tilde{O}(\eps)) \Id$, and in $\poly(n)$-time, outputs an estimate $\hat{\Sigma} \in \R^{d \times d}$ with the following guarantee. If $n \geq \tilde{O}(d^2 \log^5(1/\delta)/\eps^2)$, then with probability at least $1 - \delta$ over the draw of the original uncorrupted sample $X$, the estimate $\hat{\Sigma}$ satisfies $\norm{\Sigma^{-1/2} \hat{\Sigma} \Sigma^{-1/2} - \Id}_F \leq \tilde{O}(\eps)$.
}
\MYload{thm:frobnorm}
\end{mtheorem}

We note that
\cref{thm:frobnorm} requires that the input distribution has covariance that is close to $\Id$ in spectral norm. This is easily achieved by first running \cref{program:canonical} to derive an estimate $\hat \Sigma_0$ that is close to $\Sigma$ in spectral norm, and then running \cref{program:frobnorm} on inputs that are linearly transformed as $y \mapsto \hat \Sigma_0^{-1/2}y$. After this linear transformation, the new, ``true'' covariance $\Sigma_0^{-1/2} \Sigma \hat \Sigma_0^{-1/2}$ satisfies $(1 - \tilde{O}(\eps)) \Id \preceq \hat \Sigma_0^{-1/2} \Sigma \hat \Sigma_0^{-1/2} \preceq(1 + \tilde{O}(\eps)) \Id$.

We thus obtain the final corollary: 
\begin{corollary}[Mean and Frobenius norm covariance estimation]
\label{cor:finalalg}
There is an SoS-based algorithm that takes as input an $\epsilon$-corrupted sample of size $n$ from a Gaussian distribution with mean $\mu$ and covariance $\Sigma$ and in $\poly(n)$-time, outputs estimates $\hat{\mu} \in \R^d, \hat{\Sigma} \in \R^{d \times d}$ with the following guarantee. If $\Sigma \succeq 2^{-\poly(d)}I$, and $n \geq \tilde{O}(d^2 \log^5(1/\delta)/\eps^2)$, with probability at least $1 - \delta$ over the draw of the original uncorrupted sample $X$, the estimates $\hat{\mu}, \hat{\Sigma}$ satisfy:
\begin{enumerate}[(1)]
\item (Mean estimation) $\Norm{\Sigma^{-1/2}(\hat{\mu} - \mu)}_2 \leq \tilde{O}(\eps)$, and
\item (Covariance estimation in Frobenius norm) $\norm{\Sigma^{-1/2} \hat{\Sigma} \Sigma^{-1/2} - \Id}_F \leq \tilde{O}(\eps)$.
\end{enumerate}
In particular, $\tvdist(N(\mu, \Sigma), N(\hat{\mu}, \hat{\Sigma})) \leq \tilde{O}(\eps)$.
\end{corollary}
Thus, we obtain the same guarantee\footnote{Our formal guarantees are not explicit about $\polylog(1/\eps)$ factors in the error, and as a result, formally speaking our error bounds only match that of \cite{DiakonikolasKK016} up to $\polylog(1/\eps)$ factors. We believe that our argument in fact gets the same $\polylog(1/\eps)$ dependence, as we rely on the same concentration bounds as in \cite{DiakonikolasKK016}, which is where the $\polylog(1/\eps)$ factors arise. But, our proofs currently do not explicitly show this.} on the total variation distance as in \cite{DiakonikolasKK016}.

\subsection{A brief overview of our key idea}
We now give a high level sketch of the key idea used in our proof. First, let's briefly recap the style of analysis in~\cite{KS17,BakshiK20} by focusing on the guarantee for mean estimation. The analysis in these works utilizes the ``proofs to algorithms'' framework of algorithm design via the sum-of-squares method. The polynomial constraints in the our program (\cref{program:canonical}) encodes finding a set $X'$ of samples that intersects the input corrupted sample $Y$ in $(1-\epsilon)n$ points and has $4$th moments upper bounded in terms of the squared $2$nd moments in every direction. The analysis proceeds by using the constraints to derive, via a $O(1)$-degree SoS proof that the error in the so called \emph{Mahalanobis} norm, $\Norm{\Sigma^{-1/2}(\mu'-\mu)}_2 \leq O(\epsilon^{3/4})$. Such an inequality implies that any degree $O(1)$-pseudo-expectation $\pE$ that satisfies the constraints of our program must also satisfy all consequences obtainable via $O(1)$-degree SoS proofs. As a result, the rounded estimate $\hat{\mu} = \pE[\mu']$ also satisfies the inequality above giving us the required guarantee. 

The Mahalanobis error bound of $\epsilon^{3/4}$ here comes from the upper bound on the $4$th moments (in general, we can obtain $\sim \sqrt{t}\epsilon^{1-1/t}$ error bounds by working with upper bounds on $t$-th moments) encoded in our constraints and is polynomially off from the optimal $\tilde{O}(\epsilon)$ bound we intend to achieve when the unknown distribution is Gaussian. In fact, the analysis and the bounds obtained by~\cite{KS17,BakshiK20} are \emph{information-theoretically optimal}: for any $t$, there are two distributions with means that are $\sqrt{t}\epsilon^{1-1/t}$ far, satisfy the $2t$-th moment upper bound condition, and are $\sim \epsilon$-different in total variation distance. In particular, one can take such a pair of distributions and produce corrupted samples that are statistically indistinguishable!

At this point, one might wonder -- how can one hope the \emph{same set of constraints} to yield a tighter guarantee for Gaussians? Indeed, our analysis follows a substantially different path to use the Gaussianity of the underlying input distribution. 

\paragraph{Enter Resilience.} At a high-level, in retrospect, the key property of Gaussians that \cite{DiakonikolasKK016} exploit (which is not satisfied by distributions that constitute the ``hard examples'' above) is a certain mild anti-concentration property (that we call resilience, following~\cite{DBLP:journals/corr/SteinhardtCV17}) inferred from \emph{lower bounds} on moments of subsamples. Specifically, if $X$ is a typical i.i.d.\ sample from a $d$-dimensional, $0$-mean Gaussian distribution of size $n \gg d/\epsilon$, and $S \subseteq X$ is \emph{any} subset of size $(1-\epsilon)n$, then the empirical covariance $\Sigma_S$ of points in $S$ satisfies $\Sigma_S \geq (1-\tilde{O}(\epsilon)) \Sigma_X$. Or, equivalently, that for $\bar{S} = X \setminus S$, it holds that $\Sigma_{\bar{S}} \preceq \polylog(1/\eps) \Sigma_X$. The analysis of~\cite{KS17} can only infer (via Cauchy-Schwarz inequality) the exponentially worse bound of $\Sigma_{\bar{S}} \leq O(1/\sqrt{\epsilon}) \Sigma_X$. 

Crucially, resilience of the covariance \emph{cannot} be inferred from 4th moment upper bounds, such as those encoded by our constraints. It can indeed be inferred from an argument that relies on boundedness of $O(\log 1/\epsilon)$ moments, but if we wanted our sample $X$ to have all of its $\leq O(\log 1/\epsilon)$ moments close to that of the true distribution, we would need $d^{O(\log 1/\epsilon)}$ (in particular, superpolynomially many) samples.

A key insight of \cite{DiakonikolasKK016} is the observation that one can prove the resilience of covariance by a simple Hoeffding + union bound for a sample of size $n \sim d/\epsilon$. Notice that Hoeffding's inequality itself relies on subgaussianity of all moments of the distribution but the relevant consequence of it -- namely resilience -- can be ``seen'' in typical samples of size $\sim d/\epsilon$. 

\paragraph{Resilience is likely not efficiently certifiable.} While this is encouraging, using this property within the sum-of-squares framework poses a major issue. Notice that, a priori, verifying that a sample $X$ satisfies resilience requires an exponential search since we need to verify some property for every subset $S$ of size $(1-\epsilon)n$. Indeed, given a sample of size $n$ -- as far as we know, there is no known polynomial time algorithm to output a \emph{certificate} (whether via sum-of-squares or otherwise) of such a property. On the other hand, since the analysis style in~\cite{KS17} involves deriving a bound on the Mahalanobis distance between $\mu'$ and the true mean $\mu$, we would need a low-degree sum-of-squares certificate of resilience in order to plug it into the SoS framework. This is the key technical issue that prevented prior attempts to ``SoSize'' the argument of~\cite{DiakonikolasKK016} for mean (and more generally, covariance estimation) for Gaussians. 

\paragraph{Circumventing certificates by proving ``only in pseudo-expectation".} Our main contribution is an argument that allows us to use resilience without requiring an SoS certificate.  Notice that, though powerful and elegant, obtaining a low-degree sum-of-squares proof of a bound on the Mahalanobis distance $\Norm{\Sigma^{-1/2}(\mu'-\mu)}_2$ is overkill for our purpose! We only need the inequality \emph{after taking pseudo-expectations}. Our key idea, thus, is to directly prove a bound on $\Norm{\Sigma^{-1/2}(\pE[\mu']-\mu)}_2$ without going through low-degree sum-of-squares proofs. 

If, for a second, we pretend that pseudo-expectations are in fact actual probability distributions over solutions $X'$, then this is akin to proving an inequality on the expectation of the solution $X'$ without establishing (the considerably stronger claim) that it holds ``pointwise'' in the support of the distribution. Thus, our idea above can be summarized as attempting to prove a fact ``in pseudo-expectation'' without establishing it pointwise in the support of the ``pseudo-distribution''. 

We show that for the purpose of arguing ``after taking pseudo-expectations'',  we can in fact leverage the resilience bound discussed above. Our final argument thus derives some facts ``within low-degree sum-of-squares proof system'' -- with some technical choices that make the composition with facts ``after taking expectations'' possible. Making this work and extending to covariance estimation in spectral and then Frobenius norms requires some more technical work which, for the purpose of this overview, we omit.

To the best of our knowledge, this is the first example in the SoS proofs to algorithms framework for robust statistics where the difference between facts ``derived via low-degree SoS proofs'' vs ``proved only in pseudo-expectations'' appears to make a significant material difference to the results so obtained. We believe that this style of analysis might come in handy in future applications of the SoS method to robust statistics and more generally, problems in statistical estimation. 

\section{Preliminaries}
\label{sec:sos}
In this section, we give an overview of the sum-of-squares algorithm and state the concentration properties of Gaussians that we need for our results.
\subsection{A crash course in sum-of-squares}
We give a brief overview of the sum-of-squares (SoS) algorithm. For a more in-depth survey, see~\cite{FlemingKP19}. 

The sum-of-squares algorithm works in the standard word RAM model of computation. We assume that all numerical inputs are rational numbers represented as a pair of integers describing the numerator and the denominator. In order to measure the running time of algorithms, we will need to account for the length of the numbers that arise during the run of the algorithm. The following definition captures the size of the representations of the rational numbers:
\begin{definition}[Bit complexity]
The bit complexity of an integer $p \in \Z$ is $1+ \lceil \log_2 p \rceil$. The bit complexity of a rational number $p/q$ where $p,q \in \Z$ is the sum of the bit complexities of $p$ and $q$. 
\end{definition}

We now move to discussing the SoS algorithm. Consider a generic polynomial feasibility problem of the form
\begin{align}
    \qq{find} x \in \R^m \qq{s.t.} f_i(x) \ge 0~~\forall i,\quad g_j(x) = 0~~\forall j \label{eq:csp}
\end{align}
where $f_i$ and $g_j$ are arbitrary polynomial functions of $x$ with rational coefficients of bit complexity $B$, and the total number of constraints is $\poly(m)$. Let $\P_{m,k}$ denote the set of polynomials $p$ in $m$ variables with degree at most $k$. A  degree-$k$ pseudo-expectation is an object that mimics a real expectation $\E$ for low-degree polynomials, and is defined as follows.
\begin{definition}[Degree-$k$ pseudo-expectation]
A degree-$k$ pseudo-expectation ($k$ even) over $m$ variables is a linear operator $\pE \colon \P_{m,k} \to \R$ satisfying:
\begin{enumerate}
    \item (\emph{Normalization}) $\pE[1] = 1$, and
    \item (\emph{PSDness}) $\pE[p^2] \geq 0$ for all $p \in \P_{m, k/2}$.
\end{enumerate} 
\end{definition}
We say that the PSDness condition is satisfied with error $\tau$ if $\pE[p^2] \geq -\tau \norm{p}_2^2$ for each $p \in \P_{m, k/2}$, where $\norm{p}_2$ is the $\ell_2$-norm of the vector of coefficients of $p$.

We now define what it means for $\pE$ to (approximately) satisfy constraints.
\begin{definition}[Satisfying constraints]
For a polynomial $g$, we say that a degree-$k$ $\pE$ satisfies the constraint $\{g = 0\}$ exactly if for every polynomial $p$ of degree $\leq k-\deg(g)$, $\pE[p g_j] = 0$ and $\tau$-approximately if $|\pE[p g_j]| \leq \tau \norm{p}_2$. We say that $\pE$ satisfies the constraint $\{g \geq 0\}$ exactly if for every polynomial $p$ of degree $\leq k/2 - \deg(g)/2$, it holds that $\pE[p^2 g] \geq 0$ and $\tau$-approximately if $\pE[p^2 g] \geq -\tau \norm{p}_2^2$.  
\end{definition}
We note that in the above two definitions, the requirements on the degree
of the polynomial is such that $\pE$ is well-defined, e.g., $\pE[p g_j]$ is only well-defined when $\deg(p g_j) \leq k$. 

For intuition, it is helpful to observe that a pseudo-expectation is a relaxation of the familiar notion of expectations: it may be useful to think of pseudo-expectation as satisfying $\pE[p] = \E_{z \sim D}[p(z)]$ for some distribution $D$ over $\R^m$.
Clearly, if $D$ is a distribution over \emph{feasible} solutions of the constraints in~\eqref{eq:csp}, then $\pE$ satisfies all constraints.

We are now ready to define the {\em sum-of-squares algorithm}.
\begin{fact}[Sum-of-Squares algorithm, \cite{Shor87,parrilo2000structured,Nesterov00,Lasserre01}]
There is an algorithm, the \emph{degree-$k$ sum-of-squares algorithm}, with the following properties: The algorithm takes as input $B \in \N$, $\tau >0$, $k \in \N$, and a problem of the form~\eqref{eq:csp} with $\poly(m)$ constraints, each with rational coefficients of bit complexity $B$. If there is a degree-$k$ pseudo-distribution satisfying~\eqref{eq:csp}, then the algorithm outputs in $\poly(B, \log \frac{1}{\tau}) \cdot m^{O(k)}$ a degree-$k$ pseudo-expectation $\pE$ that $\tau$-approximately satisfies all the constraints in~\eqref{eq:csp}, and otherwise outputs ``infeasible''.
\end{fact}
For the purposes of this paper, we can set $\tau = 2^{-m}$ and $B = \poly(m)$. The ``total error'' that we will incur will be $O(\poly(m,B) 2^{-m}) = O(\poly(m) 2^{-m})$ which is negligible.

We state some basic known facts about the pseudo-expectations that we use below.
\begin{proposition}[see for e.g., \cite{BarakS16,FlemingKP19}]
\label{lem:pEfacts}
For any degree-$k$ $\pE$, the following Cauchy-Schwarz inequalities hold: 
\begin{enumerate}
    \item For any $p, q \in \P_{m,k/2}$, we have $\pE[pq]^2 \leq \pE[p^2]\pE[q^2]$.
    \item For any $p_1, \dots, p_n, q_1, \dots, q_n \in \P_{m,k/2}$ and distribution $\mathcal D$ over $[n]$, $\pE$ satisfies the polynomial inequality $\E_{i \sim \mathcal D}[p_i q_i]^2 \leq \E_{i \sim \mathcal D}[p_i^2]\E_{i \sim \mathcal D}[q_i^2]$. In particular, $(p_1 + p_2 + p_3)^2 \leq  3(p_1^2 + p_2^2 + p_3^2)$.
\end{enumerate}
\end{proposition}
\begin{definition}[SoS proofs of non-negativity]
\label{def:sosproof}
Let $h : \R^d \to \R$ be a polynomial. We say that $h$ {\em has a degree-$\ell$ SoS proof of nonnegativity} if $h = \sum_{i=1}^r h_i^2$ for some polynomials $h_i \in \P_{d, \ell/2}$. 
\end{definition}

\subsection{Resilience and certifiable subgaussianity of Gaussian moments}
\label{sec:gaussianproperties}
We now give a brief overview of the key properties of Gaussian moments that we use.

Our algorithm relies on concentration bounds of Gaussians from~\cite{DiakonikolasKK016}, which prove resilience of the first $4$ moments of the Gaussian distribution. We state the bounds for the first two moments below.
\begin{lemma}[Resilience of first and second moments; Lemmas~4.4, 4.3 in \cite{DiakonikolasKK016}]
\label{lem:resiliencemean}
Let $x_1,\dots, x_n \sim \cN(0, \Id_{d \times d})$, and $n \ge 
O((d+\log(1/\delta))/\eps^2)$.
Then with probability $1-\delta$, for all $v \in \R^d$ and vectors $a \in [0, 1]^n$ such that $\E_i a_i \ge 1 - \eps$, we have
\begin{flalign*}
&\bigabs{\E_i a_i \ip{x_i, v}} \leq \tilde{O}(\eps) \norm{v}_2 \enspace, \qq{and} \\
&\bigabs{\E_i a_i [ \ip{x_i, v}^2 - \norm{v}_2^2]} \leq \tilde{O}(\eps) \norm{v}_2^2 \enspace.
\end{flalign*}
\end{lemma}
To see the importance of \cref{lem:resiliencemean}, we note that, when combined with Proposition~2 in~\cite{SteinhardtCV18}, \cref{lem:resiliencemean} immediately yields an exponential time algorithm to robustly estimate the mean $\mu$ of a Gaussian $\cN(\mu, \Sigma)$ with \emph{known} covariance $\Sigma$, i.e., output $\hat{\mu}$ satisfying (1) in \cref{thm:main}.

The second resilience property we need is an upgrade of the resilience property of the second moment in \cref{lem:resiliencemean}, as well as the resilience of the fourth moment.
\begin{lemma}[Stronger resilience of second and resilience of fourth moments]
\label{lem:resiliencecov}
Let $x_1,\dots, x_n \sim \cN(0, \Id_{d \times d})$, and $n \ge \tilde O(d^2 \log^5(1/\delta)/\eps^2)$. Then with probability $1-\delta$, for all symmetric $P \in \R^{d \times d}$ and vectors $a \in [0, 1]^n$ such that $\E_i a_i \ge 1 - \eps$, we have
\begin{flalign*}
&\bigabs{\E_{i}a_i \ip{x_i x_i^\top - \Id, P}} \leq \tilde{O}(\eps) \cdot \norm{P}_F \enspace, \qq{and}  \\
&\bigabs{\E_{i} a_i \qty\big[\ip{x_i x_i^\top - \Id, P}^2 - 2\norm{ P }_F^2]} \leq \tilde{O}(\eps) \cdot \norm{ P }_F^2 \enspace.
\end{flalign*}
\end{lemma}
\cref{lem:resiliencecov} follows from Corollary~4.8, Lemma~5.17 and Lemma~5.21 of~\cite{DiakonikolasKK016}; we include a short proof for completeness in \cref{sec:concproof}.

We will also need slightly different forms of the above resilience results. We now state the results in the form that we need, and postpone the proof to \cref{sec:concproof}. The proofs are a straightforward (but somewhat tedious) consequence of the above results from~\cite{DiakonikolasKK016}.

\begin{lemma}
\label{lem:annoying-statements}
\MYstore{lem:annoying-statements}{
Let $x_1, \dots, x_n \sim \cN(\mu, \Sigma)$ for $\mu \in \R^d$ and $\Sigma \in \R^{d \times d}$ positive definite. Let $n$ be as in \cref{lem:resiliencecov}, and $\mu_0 = \E_i x_i$ and  $\Sigma_0 = \E_i (x_i - \mu_0) (x_i - \mu_0)^\top$ be the sample mean and covariance respectively. Let $X_{ij} = \frac{1}{2}(x_i - x_j)(x_i - x_j)^\top$, and let $a_{ij} \in [0,1]$ for $i, j \in [n]$ and $a_i$ for $i \in [n]$ be such that
\begin{enumerate}[(1)]
\item $a_{ij} = a_{ji}$ for all $i,j$,
\item $\E_{ij} a_{ij} \geq 1 - 4 \eps$, and 
\item $ \E_j a_{ij} \ge a_i (1 - 2\eps)$ for all $i$, and $a_{ij} \le a_i$ for all $i$ and $j$.
\end{enumerate}
Then, with probability $1 - O(\delta)$, for all $v \in \R^d$ and symmetric $P \in \R^{d \times d}$, we have
\begin{enumerate}
\item $\displaystyle \abs{\ip{\mu - \mu_0, v}} \leq \tilde{O}(\eps) \sqrt{v^{\top} \Sigma v} \enspace,$
\item$\displaystyle \bigabs{\E_i a_i \ip{x_i - \mu_0, v}} \leq \tilde{O}(\eps) \cdot \sqrt{v^{\top} \Sigma_0 v} \enspace,$
\item$\displaystyle \bigabs{\E_i a_i [ \ip{x_i - \mu_0, v}^2 - v^\top \Sigma_0 v]} \leq \tilde{O}(\eps) \cdot v^{\top} \Sigma_0 v \enspace.$
\item$\displaystyle \bigabs{\ip{\Sigma_0 - \Sigma, P}} \leq \tilde{O}(\eps) \norm{\Sigma^{1/2} P \Sigma^{1/2}}_F \enspace,$
\item$\displaystyle \bigabs{\E_{i} \qty\big[\ip{(x_i - \mu_0)(x_i - \mu_0)^\top - \Sigma_0, P}^2 - 2\norm{\Sigma^{1/2} P \Sigma^{1/2}}_F^2]} \leq \tilde{O}(\eps) \cdot \norm{\Sigma^{1/2} P \Sigma^{1/2}}_F^2 \enspace, $
\item$\displaystyle \bigabs{\E_{ij}a_{ij} \ip{X_{ij} - \Sigma_0, P}} \leq \tilde{O}(\eps) \cdot \norm{\Sigma^{1/2} P \Sigma^{1/2}}_F \enspace, $ and
\item$\displaystyle \bigabs{\E_{ij} a_{ij} \qty\big[\ip{X_{ij} - \Sigma_0, P}^2 - 2\norm{\Sigma^{1/2} P \Sigma^{1/2}}_F^2]} \leq \tilde{O}(\eps) \cdot \norm{\Sigma^{1/2} P \Sigma^{1/2}}_F^2 \enspace.$
\end{enumerate}}\MYload{lem:annoying-statements}
\end{lemma}

The final property we will need of Gaussians is \emph{certifiable subgaussianity}, which says that certain moment inequalities have low-degree SoS proofs.
\begin{lemma}[Certifiable fourth moments of Gaussian samples, Section~5 in~\cite{KS17}]
\label{lem:certifiable-hypercontractivity}
\MYstore{lem:certifiable-hypercontractivity}{
Let $\eps, \delta > 0$, and $n \ge \tilde O((d \log(1/\delta)/\eps)^2)$. Let $x_1, \dots, x_n \sim \cN(0, \Sigma)$ be samples from a $d$-dimensional Gaussian. Then with probability $1 - \delta$,
\begin{align}
   h(x,v) := (3 + \eps) \ip{v, \Sigma v}^2 - \E_{i \gets [n]} \ip{x_i, v}^4
\end{align}
has a degree-$4$ SoS proof of nonnegativity in $v$ (\cref{def:sosproof}).
}
\MYload{lem:certifiable-hypercontractivity}
\end{lemma}

\section{Mean and Covariance Estimation of Gaussians via SoS}
In this section, we prove \cref{thm:main,thm:frobnorm}. First, we prove \cref{thm:main}. Then, we prove \cref{thm:frobnorm} in \cref{sec:frobnorm}. We combine \cref{thm:main,thm:frobnorm} to prove \cref{cor:finalalg} in \cref{sec:finalalg}.

\subsection{Analyzing the canonical SoS program: proof of \cref{thm:main}}
\label{sec:thmmain}

We now prove \cref{thm:main}, restated below.
\restatetheorem{thm:main}
For convenience, we shall assume without loss of generality that $\eps n$ is an integer; this can be done by changing $\eps$ by at most a constant factor.

The canonical degree-$\degfin$ SoS relaxation of \cref{program:canonical} outputs a degree-$\degfin$ pseudo-expectation $\pE$ in the variables $x'_1, \dots, x'_n \in \R^d$, $w_1, \dots, w_n \in \R$, satisfying the constraints of \cref{program:canonical}, if such a $\pE$ exists. The estimates produced by the algorithm are $\hat{\mu} := \pE[\mu']$ and $\hat{\Sigma} := \pE[\Sigma']$.

Let $x_1, \dots, x_n \sim \cN(\mu, \Sigma)$. Let $\mu_0 = \E_i x_i$ be the sample mean, and let $\Sigma_0 = \E_i (x_i - \mu_0)(x_i - \mu_0)^{\top}$ be the sample covariance. Fix $\eps \in (0,1)$, and let $y_1, \dots, y_n$ be an $\eps$-corruption of $x_1, \dots, x_n$.

By \cref{lem:annoying-statements}, with probability $1 - \delta$, the following inequalities hold for any $a_1, \dots, a_n \in [0,1]$ with $\sum_i a_i \geq (1 - 2 \eps)n$ and $v \in \R^d$:
\begin{flalign}
&\abs{\ip{\mu - \mu_0, v}} \leq \tilde{O}(\eps) \sqrt{v^{\top} \Sigma v} \enspace, \label{eq:mean1} \\
&\bigabs{\E_i a_i \ip{x_i - \mu_0, v}} \leq \tilde{O}(\eps) \cdot \sqrt{v^{\top} \Sigma_0 v} \enspace, \label{eq:mean2} \\
 &\bigabs{\E_i a_i [ \ip{x_i - \mu_0, v}^2 - v^\top \Sigma_0 v]} \leq \tilde{O}(\eps) \cdot v^{\top} \Sigma_0 v \enspace. \label{eq:mean3}
\end{flalign}

Next, we let $X_{ij} := \frac{1}{2}(x_i - x_j)(x_i - x_j)^{\top}$, for any $i,j \in [n]$.
Let $T \subseteq [0,1]^{n^2}$ denote the set of $(a_{ij})_{i,j \in [n]}$ such that:
\begin{enumerate}[(1)]
\item $a_{ij} = a_{ji}$ for all $i,j$,
\item $\sum_{i, j = 1}^n a_{ij} \geq (1 - 4 \eps)n$, and 
\item there exist $a_1, \dots, a_n \in [0,1]$ such that $ a_i \ge \E_j a_{ij} \ge a_i (1 - 2\eps)$ for all $i \in [n]$.
\end{enumerate}
By \cref{lem:annoying-statements} (setting $P = vv^{\top}$), with probability $1 - \delta$, the following inequalities hold for any $(a_{ij})_{i,j \in [n]} \in T$:
\begin{flalign}
&\bigabs{v^{\top}(\Sigma_0 - \Sigma)v} \leq \tilde{O}(\eps) v^{\top} \Sigma_0 v \label{eq:spec1} \enspace, \\
&\bigabs{\E_{i} [(\ip{x_i - \mu_0, v}^2 - v^{\top}\Sigma_0 v)^2 - 2(v^{\top} \Sigma_0 v)^2]} \leq \tilde{O}(\eps) \cdot (v^{\top} \Sigma_0 v)^2 \enspace, \label{eq:spec2}\\
&\bigabs{\E_{ij}a_{ij} [v^{\top}(X_{ij} - \Sigma_0) v]} \leq \tilde{O}(\eps) v^{\top} \Sigma_0 v \enspace, \label{eq:spec3} \\
&\bigabs{\E_{ij} a_{ij} [(v^{\top} (X_{ij} - \Sigma_0) v)^2 - 2(v^{\top} \Sigma_0 v)^2]} \leq \tilde{O}(\eps) (v^{\top} \Sigma_0 v)^2 \label{eq:spec4} \enspace.
\end{flalign}
We proceed with the rest of the proof, assuming that the above resilience conditions hold. From this point on, we will no longer need to use the randomness of the $x_i$'s.

\subsubsection{Feasibility} Let us now argue that the constraints in \cref{program:canonical} are feasible. Set $x'_i = x_i$ for each $i \in [n]$, and let $w_i = 1$ if $y_i = x_i$ and $0$ otherwise. Constraints (1), (2), (3) of \cref{program:canonical} are clearly satisfied, so it remains to argue that constraint (4) is satisfied. By \cref{eq:spec2} (with $a_i = 1$ for all $i$) constraint (4) is satisfied. Hence, the constraints in \cref{program:canonical} are feasible. In particular, \cref{program:canonical} will output in $\poly(n)$ time a degree-$\degfin$ pseudo-expectation $\pE$ in the variables $x'_1, \dots, x'_n$, $w_1, \dots, w_n$, satisfying the constraints of \cref{program:canonical}. From this point on, we shall think of the pseudo-expectation $\pE$ as fixed.

In light of the above, we summarize our notation in the box below.
  \begin{mdframed}
  \textbf{Notation:}
\begin{itemize}
\item $\mu, \Sigma$, the true mean/covariance of the Gaussian $\cN(\mu, \Sigma)$
\item $\mu_0, \Sigma_0$, the sample mean/covariance of the true samples $x_1, \dots, x_n$
\item $\mu', \Sigma'$, the SoS variables for the mean/covariance
\item $\hat{\mu} = \pE[\mu'], \hat{\Sigma} = \pE[\Sigma']$, the estimates for the mean/covariance
\item $y_1, \dots, y_n$, the $\eps$-corruption of the true samples $x_1, \dots, x_n$
\item $x'_1, \dots, x'_n$, the SoS variables for the samples
\item $w_1, \dots, w_n$, the SoS variables for the indicators $\1(x'_i = y_i)$
\end{itemize}
  \end{mdframed}
  
\subsubsection{Guarantees for the mean}
We now analyze the estimate $\hat{\mu} := \pE[\mu'] = \pE[\E_i x'_i]$ for the mean $\mu$, where $\pE$ is a degree-$\degfin$ pseudo-expectation satisfying the constraints in \cref{program:canonical}. 
We need to show that $\hat{\mu}$ satisfies $\abs{\ip{\hat{\mu} - \sampmean, v}} \leq \tilde{O}(\eps) \sqrt{v^\top \Sigma v}$.

The key ingredient in the proof is the following lemma, which we prove in \cref{sec:technicallem}.
\begin{lemma}\label{lem:technicallem-mean}
Let $x_1, \dots, x_n \in \R^d$.
Suppose that there is some $\Sigma \in \R^{d \times d}$ such that for all $v \in \R^d$ and $a \in [0,1]^d$ with $\sum_{i = 1}^n a_i \geq (1 - 2\eps)n$, we have
\begin{align}
&\abs{\E_i a_i \ip{x_i - \sampmean, v}} \leq \tilde{O}(\eps) \cdot \sqrt{v^{\top} \Sigma_0 v}\qq{and}
\abs{\E_i a_i [ \ip{x_i - \sampmean, v}^2 - v^{\top} \Sigma_0 v ]} \leq \tilde{O}(\eps) v^{\top} \Sigma_0 v \enspace.
\end{align}
Let $y_1, \dots, y_n$ be any $\eps$-corruption of $x_1, \dots, x_n$, let $\pE$ be a degree-$\deghalf$ pseudo-expectation in the variables $x'_1, \dots, x'_n \in \R^d$ and $w_1, \dots, w_n \in \R$. Let $\mu' = \E_i x'_i$. Suppose that
\begin{enumerate}[(1)]
\item $\pE$ satisfies $w_i^2 = w_i$ for every $i \in [n]$,
\item $\pE$ satisfies $\sum_{i = 1}^n w_i = (1 - \eps) n$,
\item $\pE$ satisfies $w_i x'_i = w_i y_i$ for every $i \in [n]$,
\item $\pE[\E_i \ip{x'_i - \mu', v}^2] \leq v^{\top} \hat{\Sigma} v$ for every $v \in \R^d$
\end{enumerate}
Then, for every $v \in \R^d$, it holds that:
\begin{flalign*}
&\abs{\ip{\hat{\mu} - \sampmean, v}} \leq \tilde{O}(\eps)\sqrt{v^{\top} \Sigma_0 v} + \sqrt{ O(\eps) \cdot v^{\top}(\hat{\Sigma} - \Sigma_0) v + \tilde{O}(\eps^2) v^{\top} (\hat{\Sigma} + \Sigma_0)v} \enspace.
\end{flalign*}
\end{lemma}

We now finish the proof, assuming \cref{lem:technicallem-mean}. We first observe that the hypotheses of \cref{lem:technicallem-mean} hold. Indeed, the two resilience conditions of \cref{lem:technicallem-mean} follow by \cref{eq:mean2,eq:mean3}. Second, $\pE$ is a degree-$\degfin$ pseudo-expectation (and so is also degree-$\deghalf$) with the required properties: (1) -- (3) clearly hold, and (4) follows from the definition of $\Sigma'$, as $\hat{\Sigma} = \pE[\Sigma']$.
As the hypotheses of \cref{lem:technicallem-mean} are satisfied, we thus conclude that
\begin{equation}
\abs{\ip{\hat{\mu} - \sampmean, v}} \leq \tilde{O}(\eps)\sqrt{v^{\top} \Sigma_0 v} + \sqrt{ O(\eps) \cdot v^{\top}(\hat{\Sigma} - \Sigma_0) v + \tilde{O}(\eps^2) v^{\top} (\hat{\Sigma} + \Sigma_0)v} \enspace. \label{eq:meancalc}
\end{equation}
Suppose that the estimate for the covariance $\hat{\Sigma}$ satisfies the desired conclusion, i.e., that $\abs{v^{\top}(\hat{\Sigma} - \Sigma) v} \leq \tilde{O}(\eps) v^{\top} \Sigma v$ for all $v \in \R^d$ (we will prove this next). Then, \cref{eq:meancalc,eq:spec1} imply that 
\begin{equation*}
\abs{\ip{\hat{\mu} - \sampmean, v}} \leq \tilde{O}(\eps) \sqrt{v^{\top} \Sigma v} \enspace.
\end{equation*}
Finally, by \cref{eq:mean1}, we conclude that
\begin{equation*}
\abs{\ip{\hat{\mu} - \mu, v}} \leq \tilde{O}(\eps) \sqrt{v^{\top} \Sigma v} \enspace,
\end{equation*}
assuming that $\hat{\Sigma}$ satisfies its desired property. By choosing $v$ appropriately, this implies (1) in \cref{thm:main}.

\subsubsection{Spectral guarantees on the covariance} We now analyze the estimate $\hat{\Sigma} := \pE[\Sigma']$ for the covariance, where $\pE$ is a degree-$\degfin$ pseudo-expectation satisfying \cref{program:canonical}. First, we observe that the polynomial $\Sigma' := \E_i (x'_i - \mu')(x'_i - \mu')^{\top}$ is equal to $\E_{ij} X'_{ij}$ where $X'_{ij} := \frac{1}{2}(x'_i - x'_j)(x'_i - x'_j)^{\top}$, and similarly we also have $\Sigma_0 = \E_{ij} X_{ij}$, where $X_{ij} := \frac{1}{2}(x_i - x_j)(x_i - x_j)^{\top}$.

Let $T \subseteq [0,1]^{n^2}$ denote the set of $(a_{ij})_{i,j \in [n]}$ such that:
\begin{enumerate}[(1)]
\item $a_{ij} = a_{ji}$ for all $i,j$,
\item $\sum_{i, j = 1}^n a_{ij} \geq (1 - 4 \eps)n$, and 
\item there exist $a_1, \dots, a_n \in [0,1]$ such that $\E_j a_{ij} \ge a_i (1 - 2\eps)$ for all $i$, and $a_{ij} \le a_i$ for all $i$ and $j$.
\end{enumerate}
The key ingredient here is the following lemma, which is very similar to \cref{lem:technicallem-mean} that appeared in the case of mean estimation.
\begin{lemma}
\label{lem:technicallem-spec}
Let $X_1, \dots, X_{n^2} \in \R^{d \times d}$, and let $\Sigma_0 := \E_{ij} X_{ij}$.
Let $T \subseteq [0,1]^{n^2}$. Suppose that, for all $v \in \R^{d}$ and $a \in T$ such that $\sum_{i,j} a_{ij} \geq (1-4\eps) n^2$, we have
\begin{align}
&\abs{\E_{ij} a_{ij} v^{\top}(X_{ij} - \Sigma_0)v} \leq \tilde{O}(\eps) \cdot v^{\top} \Sigma_0 v \qq{and} \abs{\E_{ij} a_{ij} [ (v^{\top}(X_{i,j} - \Sigma_0)v)^2 - 2(v^{\top} \Sigma_0 v)^2 ]} \leq \tilde{O}(\eps) (v^{\top} \Sigma_0 v)^2 \enspace.
\end{align}
Let $Y_1, \dots, Y_{n^2}$ be any $(2\eps - \eps^2)$-corruption of $X_1, \dots, X_{n^2}$, let $\pE$ be a degree-$\deghalf$ pseudo-expectation in the variables $X'_1, \dots, X'_{n^2} \in \R^{d \times d}$ and $w_1, \dots, w_{n^2} \in \R$. Let $\Sigma' = \E_{ij} X'_{ij}$. Suppose that
\begin{enumerate}[(1)]
\item $\pE$ satisfies $w_{ij}^2 = w_{ij}$ for every $i,j \in [n]$,
\item $\pE$ satisfies $\sum_{i,j = 1}^N w_{ij} = (1 - \eps)^2 n^2$,
\item $\pE$ satisfies $w_{ij} X'_{ij} = w_{ij} Y_{ij}$ for every $i,j \in [n]$,
\item $\pE[\E_{ij} (v^{\top}(X'_{ij} - \Sigma')v)^2] \leq (2 + \tilde{O}(\eps)) \pE[(v^{\top} \Sigma' v)^2]$ for every $ v \in \R^d$, and
\item $a \in T$, where $a$ is the vector with $a_{ij} := \pE[w_{ij}] \1(X_{ij} = Y_{ij})$ for each $i,j \in [n]$.
\end{enumerate}

Then, for every $v \in \R^d$, the following hold:
\begin{flalign*}
&\pE (v^{\top}(\Sigma' - \Sigma_0) v)^2 \leq O(\eps) (\pE (v^{\top} \Sigma' v)^2 + (v^{\top} \Sigma_0 v)^2 ) \enspace,\\
&\abs{\ip{\hat{\Sigma} - \Sigma_0, v}} \leq \tilde{O}(\eps)v^{\top} \Sigma_0 v + \sqrt{ \pE [\E_{ij}[ (1 - w'_{ij}) \cdot v^{\top}(X'_{ij}  - \Sigma_0) v]^2]} \enspace,
\end{flalign*}
where $w'_{ij} := w_{ij} \1(X_{ij} = Y_{ij})$, $\hat{\Sigma} := \pE[\Sigma']$, and
\begin{flalign*}
 &\pE [\E_{ij}[ (1 - w'_{ij}) \cdot v^{\top}(X'_{ij}  - \Sigma_0) v]^2] \leq O(\eps) \cdot (  \pE (v^{\top} \Sigma' v)^2 - (v^{\top} \Sigma_0 v)^2) + \tilde{O}(\eps^2) \cdot (\pE (v^{\top} \Sigma' v)^2 + (v^{\top} \Sigma_0 v)^2) \enspace.
\end{flalign*}
\end{lemma}
As before, we postpone the proof of \cref{lem:technicallem-spec} to \cref{sec:technicallem}, and use it to finish the proof.

We apply \cref{lem:technicallem-spec} as follows. First, we note that $\Sigma_0$ defined in \cref{lem:technicallem-spec} is the same as the sample mean $\Sigma_0$. Next, let $T$ be the subset of vectors $(a_{ij})_{i,j \in [n]}$ defined earlier. We see that \cref{eq:spec3,eq:spec4} imply that the $X_{ij}$'s defined satisfy the resilience conditions in \cref{lem:technicallem-spec}.

Now, we let $Y_{ij} = \frac{1}{2}(y_i - y_j)(y_i - y_j)^{\top}$, and let $X'_{ij} = \frac{1}{2}(x'_i - x'_j)(x'_i - x'_j)^{\top}$. We note that $\Sigma'$ as defined in \cref{lem:technicallem-spec} is the same polynomial as $\Sigma'$ defined earlier.
We observe that the $Y_{ij}$'s must be a $(2\eps - \eps^2)$-corruption of the $X_{ij}$'s. Moreover, if we let $w_{ij} := w_i w_j$, then the pseudo-expectation defined by $\pE$ on the polynomials $X'_{ij}$ and $w_{ij}$ is a degree-$\deghalf$ pseudo-expectation, and additionally satisfies properties (1) -- (3). To see that (4) holds, we observe the following polynomial equality:
\begin{flalign*}
\E_{ij} (v^{\top}(X'_{ij} - \Sigma')v)^2] = \frac{1}{2} (\E_i \ip{x'_i - \mu', v}^4 + (v^{\top} \Sigma' v)^2) \enspace.
\end{flalign*}
Combining with constraint (4) in \cref{program:canonical} and taking pseudo-expectations shows that property (4) holds.

Finally, property (5) in \cref{lem:technicallem-spec} holds, as $(a_{ij})_{i,j \in [n]} \in T$ because it satisfies the required properties with respect to the vector $a_i = \pE[w_i] \1(x_i = y_i)$ for each $i$.

Thus, by \cref{lem:technicallem-spec}, we have
\begin{flalign*}
\abs{\ip{\hat{\Sigma} - \Sigma_0, v v^{\top}}} \leq \tilde{O}(\eps) \cdot v^{\top} \Sigma_0 v + \sqrt{R} \enspace,
\end{flalign*}
where
\begin{flalign*}
R := \pE[\E_i[(1 - w'_{ij}) \cdot v^{\top}(X_{ij}  - \Sigma_0) v]^2] \leq O(\eps) \cdot (\pE[(v^{\top} \Sigma' v)^2] - (v^{\top} \Sigma_0 v)^2 ) + \tilde{O}(\eps^2) \cdot (\pE[(v^{\top} \Sigma' v)^2] + (v^{\top} \Sigma_0 v)^2) \enspace.
\end{flalign*}
Write $\Sigma' = A + B$, where $B = \E_{ij} (1 - w'_{ij}) X'_{ij}$ and $A = \E_{ij} w'_{ij} X'_{ij} = \E_{ij} w'_i X_{ij}$; the latter equality holds because the following polynomial equalities are satisfied by $\pE$:
\begin{flalign*}
&w'_{ij} X'_{ij} = w_{i}w_j  \1(x_i = y_i) \1(x_j = y_j) \cdot \frac{1}{2}(x'_i - x'_j)(x'_i - x'_j)^{\top} \\
&=  w_{i}w_j  \1(x_i = y_i) \1(x_j = y_j) \cdot \frac{1}{2}(y_i - y_j)(y_i - y_j)^{\top} = w_{i}w_j  \1(x_i = y_i) \1(x_j = y_j) \cdot \frac{1}{2}(x_i - x_j)(x_i - x_j)^{\top} \enspace.
\end{flalign*}
Additionally, let $A_v := v^{\top} A v$ and $B_v := v^{\top} B v$, and $\Sigma_v = v^{\top} \Sigma_0 v$. We have that
\begin{flalign*}
&\pE[A_v^2] = \pE[(\E_{ij} w'_{ij} v^{\top} X_{ij} v)^2] = \E_{i_1, j_1} \E_{i_2, j_2} \pE[w'_{i_1 j_1} w'_{i_2 j_2}] \cdot v^{\top} X_{i_1 j_1} v \cdot  v^{\top} X_{i_2 j_2} v \\
&\leq \E_{i_1, j_1} \E_{i_2, j_2} \sqrt{\pE[w'_{i_1,j_1}] \pE[w'_{i_2,j_2}]} \cdot v^{\top} X_{i_1 j_1} v \cdot  v^{\top} X_{i_2 j_2} v \ \text{(as $\pE[{w'_{ij}}^2] = \pE[w'_{ij}]$)} \\
&= \pE[\E_{i,j} \sqrt{\pE[w'_{ij}]}  v^{\top} X_{i j} v]^2\leq (1 + \tilde{O}(\eps)) \Sigma_v^2 \ \text{(by \cref{eq:spec3} applied to $a_{ij} = \sqrt{\pE[w'_{ij}]}$)}
\end{flalign*}

We now bound $R$. In this notation, we have
\begin{flalign}
\label{eq:eq1}
 R = \pE[(B_v - \E_{ij}[1 - w'_{ij}] \cdot \Sigma_v)^2] \leq O(\eps) \cdot (\pE[(A_v + B_v)^2] - \Sigma_v^2 ) + \tilde{O}(\eps^2) \cdot (\pE[(A_v + B_v)^2] +  \Sigma_v^2) \enspace.
 \end{flalign}
First, we have
 \begin{flalign}
 \label{eq:eq2}
 \pE[(B_v - \E_{ij}[1 - w'_{ij}] \cdot \Sigma_v)^2] = \pE[B_v^2 + \E_{ij}[1 - w'_{ij}]^2 \cdot \Sigma^2_v - 2 B_v \E_{ij}[1 - w'_{ij}] \cdot \Sigma_v] \geq  \pE[B_v^2] - 4 \eps \Sigma_v \pE[ B_v] \enspace,
 \end{flalign}
 as $\Sigma_v \geq 0$ and $\pE$ satisfies $B_v \geq 0$ because $B_v$ is a sum-of-squares polynomial. As $\pE[A_v^2] \leq (1 + \tilde{O}(\eps)) \Sigma_v^2$ and $\pE[A_v B_v]^2 \leq \pE[A_v^2] \pE[B_v^2]$ (by \cref{lem:pEfacts}), it follows that
 \begin{flalign}
 \label{eq:eq3}
\pE[(A_v + B_v)^2] \leq \pE[B_v^2] + 2\sqrt{\pE[A_v^2] \pE[B_v^2]} + (1 + \tilde{O}(\eps)) \Sigma_v^2 \leq \pE[B_v^2] + 2\Sigma_v \sqrt{\pE[B_v^2]} + (1 + \tilde{O}(\eps)) \Sigma_v^2 \enspace.
 \end{flalign}
Combining \cref{eq:eq1,eq:eq2,eq:eq3} thus yields
 \begin{flalign*}
 &\pE[B_v^2] - 4 \eps \Sigma_v \pE[ B_v] \leq R \leq O(\eps) (\pE[B_v^2] + 2\Sigma_v \sqrt{\pE[B_v^2]} + \tilde{O}(\eps) \Sigma_v^2) +  \tilde{O}(\eps^2) \cdot \Sigma_v^2\enspace.
 \end{flalign*}
 Rearranging, applying $\pE[B_v] \le \sqrt{\pE[B_v^2]}$, and solving for $\pE[B_v^2]$ yields
 \begin{flalign*}
  \pE[B_v^2] &\leq  \tilde{O}(\eps^2) \cdot \Sigma_v^2 \\
\implies R &\leq O(\eps) (\pE[B_v^2] + 2 \Sigma_v  \sqrt{\tilde O(\eps^2) \Sigma_v^2} + \tilde O(\eps) \Sigma_v^2) = \tilde{O}(\eps^2) \Sigma_v^2 \\
 \implies \abs{v^{\top}(\hat{\Sigma} - \Sigma_0)v} &\leq \tilde{O}(\eps) \cdot v^{\top} \Sigma_0 v + \sqrt{R} = \tilde{O}(\eps) \cdot v^{\top} \Sigma_0 v \enspace.
 \end{flalign*}
 This is the desired spectral norm guarantee, only with $\Sigma_0$ in place of $\Sigma$. Using \cref{eq:spec1} and the triangle inequality, we have $\abs{v^{\top}(\hat{\Sigma} - \Sigma) v} \leq \tilde{O}(\eps) v^{\top} \Sigma v$, and so we thus have the desired spectral norm guarantee. This finishes the proof, as we have shown that $\hat{\mu}$ satisfies its desired property, assuming that $\hat{\Sigma}$ has this property.
 
\subsection{Relative Frobenius guarantees on the covariance: proof of \cref{thm:frobnorm}}
\label{sec:frobnorm}
We now prove \cref{thm:frobnorm}, restated below.
\restatetheorem{thm:frobnorm}
Let $x_1, \dots, x_n \sim \cN(\mu, \Sigma)$, where $(1 - \tilde{O}(\eps)) \Id \preceq \Sigma \preceq (1 + \tilde{O}(\eps)) \Id$. Fix $\eps \in (0,1)$, and let $y_1, \dots, y_n$ be an $\eps$-corruption of $x_1, \dots, x_n$. Let $\mu_0 = \E_i x_i$ be the sample mean, and let $\Sigma_0 = \E_i (x_i - \mu_0)(x_i - \mu_0)^{\top}$ be the sample covariance.

We observe that for every symmetric $P \in \R^{d \times d}$, it holds that
\begin{equation}
\abs{\norm{P}_F - \norm{\Sigma^{1/2} P \Sigma^{1/2}}_F} \leq \tilde{O}(\eps) \min(\norm{P}_F, \norm{\Sigma^{1/2} P \Sigma^{1/2}}_F) \enspace, \label{eq:frob0}
\end{equation}
as $(1 - \tilde{O}(\eps))\Id \preceq \Sigma \preceq (1 + \tilde{O}(\eps))\Id$, using the standard inequality $\norm{A B}_F \leq \norm{A}_2 \norm{B}_F$.

For each $i,j \in [n]$, let $X_{ij} := \frac{1}{2}(x_i - x_j)(x_i - x_j)^{\top}$.
Let $T \subseteq [0,1]^{n^2}$ denote the set of $(a_{ij})_{i,j \in [n]}$ such that:
\begin{enumerate}[(1)]
\item $a_{ij} = a_{ji}$ for all $i,j$,
\item $\sum_{i, j = 1}^n a_{ij} \geq (1 - 4 \eps)n$, and 
\item there exist $a_1, \dots, a_n \in [0,1]$ such that $\E_j a_{ij} \ge a_i (1 - 2\eps)$ for all $i$, and $a_{ij} \le a_i$ for all $i$ and $j$.
\end{enumerate}
By \cref{lem:annoying-statements}, with probability $1 - \delta$ the following hold for any $(a_{ij}) \in T$ and symmetric $P \in \R^{d \times d}$:
\begin{flalign}
&\bigabs{\ip{\Sigma_0 - \Sigma, P}} \leq \tilde{O}(\eps) \norm{\Sigma^{1/2} P \Sigma^{1/2}}_F \label{eq:frob1} \enspace, \\
&\bigabs{\E_{i} \qty\big[\ip{(x_i - \mu_0)(x_i - \mu_0)^\top - \Sigma_0, P}^2 - 2\norm{\Sigma^{1/2} P \Sigma^{1/2}}_F^2]} \leq \tilde{O}(\eps) \cdot \norm{\Sigma^{1/2} P \Sigma^{1/2}}_F^2 \enspace, \label{eq:frob2}\\
&\bigabs{\E_{ij}a_{ij} [ \ip{X_{ij}, P} - \ip{\Sigma_0, P}]} \leq \tilde{O}(\eps) \cdot \norm{\Sigma^{1/2} P \Sigma^{1/2}}_F \enspace, \label{eq:frob3} \\
&\bigabs{\E_{ij} a_{ij} \qty\big[\ip{X_{ij} - \Sigma_0, P}^2 - 2\norm{\Sigma^{1/2} P \Sigma^{1/2}}_F^2]} \leq \tilde{O}(\eps) \cdot \norm{\Sigma^{1/2} P \Sigma^{1/2}}_F^2 \label{eq:frob4} \enspace.
\end{flalign}
From the above, feasibility is simple: set $x'_i = x_i$ for all $i$, $w_i = \1(x_i = y_i)$, and observe that now $\mu' = \mu_0$, $\Sigma' = \Sigma_0$, and constraint (4) in \cref{program:frobnorm} is satisfied by \cref{eq:frob2} as $\norm{\Sigma^{1/2} P \Sigma^{1/2}}_F \leq (1 + \tilde{O}(\eps))\norm{P}_F$. Thus, \cref{program:frobnorm} will output in $\poly(n)$ time a degree-$\degfin$ pseudo-expectation $\pE$ satisfying the constraints in \cref{program:frobnorm}.

\parhead{Covariance estimation in Frobenius norm.} We now analyze the output $\hat{\Sigma} := \pE[\Sigma']$ of the algorithm. We observe that $\E_{ij} X_{ij}$ is equal to the sample covariance $\Sigma_0$. Let $Y_{ij} := \frac{1}{2}(y_{i} - y_j)(y_i - y_j)^{\top}$, and let $X'_{ij} := \frac{1}{2}(x'_i - x'_j)(x'_i - x'_j)^{\top}$. Similarly, we have that the SoS variable $\Sigma' := \E_i (x'_i - \mu')(x'_i - \mu')^{\top}$ is equal to $\E_{ij} X'_{ij}$.

The key ingredient in the proof is the following technical lemma, which we prove in \cref{sec:technicallem}. This lemma is similar to \cref{lem:technicallem-mean,lem:technicallem-spec}.
\begin{lemma}
\label{lem:technicallem-frob}
Let $X_1, \dots, X_{n^2} \in \R^{d \times d}$, and let $\Sigma_0 := \E_{ij} X_{ij}$.
Let $T \subseteq [0,1]^{n^2}$. Suppose that, for all symmetric $P \in \R^{d \times d}$ and $a \in T$, we have
\begin{align}
&\abs{\E_{ij} a_{ij} \ip{X_{ij} - \Sigma_0, P}} \leq \tilde{O}(\eps) \cdot \norm{P}_F \qq{and}
\abs{\E_{ij} a_{ij} [ \ip{X_{ij} - \Sigma_0, P}^2 - 2\norm{P}_F^2]} \leq \tilde{O}(\eps) \norm{P}_F^2 \enspace.
\end{align}

Let $Y_1, \dots, Y_{n^2}$ be any $(2\eps - \eps^2)$-corruption of $X_1, \dots, X_{n^2}$, and let $\pE$ be a degree-$\deghalf$ pseudo-expectation in the variables $X'_1, \dots, X'_{n^2} \in \R^{d \times d}$ and $w_1, \dots, w_{n^2} \in \R$. Let $\Sigma' = \E_{ij} X'_{i,j}$. Suppose that
\begin{enumerate}[(1)]
\item $\pE$ satisfies $w_{ij}^2 = w_{ij}$ for every $i,j \in [n]$,
\item $\pE$ satisfies $\sum_{i,j = 1}^n w_{ij} = (1 - \eps)^2 n^2$,
\item $\pE$ satisfies $w_{ij} X'_{ij} = w_{ij} Y_{ij}$ for every $i,j \in [n]$,
\item $\pE[\E_{ij} \ip{X'_{ij} - \Sigma', v}^2] \leq (2 + \tilde{O}(\eps)) \norm{P}_F^2$ for every symmetric $P \in \R^{d \times d}$, and
\item $a \in T$, where $a$ is the vector with $a_{ij} := \pE[w_{ij}] \1(X_{ij} = Y_{ij})$ for each $i,j \in [n]$.
\end{enumerate}

Then, for every symmetric $P \in \R^{d \times d}$, it holds that
\begin{flalign*}
&\abs{\ip{\hat{\Sigma} - \Sigma_0, P}} \leq \tilde{O}(\eps)\norm{P}_F \enspace,
\end{flalign*}
where $\hat{\Sigma} = \pE[\Sigma']$.
\end{lemma}
We now apply \cref{lem:technicallem-frob}. We observe that by \cref{eq:frob3,eq:frob4}, and using \cref{eq:frob0}, the resilience condition in \cref{lem:technicallem-frob} is satisfied by the $X_{ij}$'s. We also observe that the pseudo-expectation $\pE$, in the variables $X'_{ij}$ and $w_{ij}$ with $w_{ij} := w_i w_j$, is a degree-$\deghalf$ pseudo-expectation, and trivially satisfies properties (1) -- (3). Property (4) follows as $\pE$ satisfies constraint (4) in \cref{program:frobnorm} and the following polynomial equality holds:
\begin{flalign*}
\E_{ij} \ip{X'_{ij} - \Sigma', v}^2 = \frac{1}{2} \E_{i} \ip{(x'_i - \mu')(x'_i - \mu')^{\top}, P}^2 + \frac{1}{2} \ip{\Sigma', v}^2 \enspace.
\end{flalign*}
Property (5) follows by using the vector $a$ with $a_i = \pE[w_i] \1(x_i = y_i)$ to show membership of $(a_{ij})_{i,j \in [n]}$ in $T$.

We thus have by \cref{lem:technicallem-frob} that $\abs{\ip{\hat{\Sigma} - \Sigma_0,P}} \leq \tilde{O}(\eps) \norm{P}_F$ for all symmetric $P \in \R^{d \times d}$. Using \cref{eq:frob1}, it follows that $\abs{\ip{\hat{\Sigma} - \Sigma,P}} \leq \tilde{O}(\eps) \norm{P}_F$.
Hence,
\begin{flalign*}
&\abs{\ip{\Sigma^{-1/2}\hat{\Sigma}\Sigma^{-1/2} - \Id, P}} = \abs{\ip{\hat{\Sigma} - \Sigma, \Sigma^{-1/2}P\Sigma^{-1/2}}} \\
&\leq \tilde{O}(\eps)\norm{\Sigma^{-1/2} P \Sigma^{-1/2}}_F \leq \tilde{O}(\eps) \norm{P}_F \enspace.
\end{flalign*}
Setting $P = \frac{\Sigma^{-1/2} \hat{\Sigma} \Sigma^{-1/2} - \Id}{\norm{\Sigma^{-1/2} \hat{\Sigma} \Sigma^{-1/2} - \Id}_F}$, we conclude that 
\begin{flalign*}
\norm{\Sigma^{-1/2} \hat{\Sigma} \Sigma^{-1/2} - \Id}_F \leq \tilde{O}(\eps) \enspace,
\end{flalign*}
as required.

\section{A Generic Estimation Lemma}
\label{sec:technicallem}
\cref{lem:technicallem-mean,lem:technicallem-spec,lem:technicallem-frob} are special cases of a generic technical result, which we now state and prove.
\begin{lemma}
\label{lem:technicallem}
Let $x_1, \dots, x_n \in \R^d$, and let $\sampmean := \E_i x_i$. Let $V(\sampmean,v)$ for $v \in \R^d$ be a degree-$2$ polynomial in $\sampmean$, and let $S \subseteq \R^d$ be a set such that $V(\sampmean, v) \geq 0$ for all $v \in S$ and $\sampmean \in \R^d$.

Let $T \subseteq [0,1]^n$. Suppose that, for all $v \in \R^d$ and $a \in T$ such that $\sum_i a_i \ge (1-\eps) n$, we have
\begin{align}
&\abs{\E_i a_i \ip{x_i - \sampmean, v}} \leq \tilde{O}(\eps) \cdot \sqrt{V(\sampmean, v)}\qq{and}
\abs{\E_i a_i [ \ip{x_i - \sampmean, v}^2 - V(\sampmean, v) ]} \leq \tilde{O}(\eps) V(\sampmean, v) \enspace. \label{eq:robustness-assumption}
\end{align}

Let $y_1, \dots, y_n$ be any $\eps$-corruption of $x_1, \dots, x_n$, let $\pE$ be a degree-$\deghalf$ pseudo-expectation in the variables $x'_1, \dots, x'_n \in \R^d$ and $w_1, \dots, w_n \in \R$. Let $\mu' = \E_i x'_i$. Suppose that
\begin{enumerate}[(1)]
\item $\pE$ satisfies $w_i^2 = w_i$ for every $i \in [n]$,
\item $\pE$ satisfies $\sum_{i = 1}^n w_i = (1 - \eps) n$,
\item $\pE$ satisfies $w_i x'_i = w_i y_i$ for every $i \in [n]$,
\item $\pE[\E_i \ip{x'_i - \mu', v}^2] \leq (1 + \tilde{O}(\eps)) \pE[V(\mu', v)]$ for every $ v \in S$, and
\item $a \in T$, where $a$ is the vector with $a_i := \pE[w_i] \1(x_i = y_i)$ for each $i \in [n]$.
\end{enumerate}

Then, for every $v \in S$, the following hold:
\begin{flalign*}
&\pE \ip{\mu' - \sampmean, v}^2 \leq O(\eps) (\pE V(\mu', v) + V(\sampmean,v)) \enspace,\\
&\abs{\ip{\hat{\mu} - \sampmean, v}} \leq \tilde{O}(\eps)\sqrt{V(\sampmean, v)} + \sqrt{ \pE [\E_i[ (1 - w'_i) \ip{x'_i  - \sampmean, v}]^2]} \enspace,
\end{flalign*}
where $\hat{\mu} := \pE[\mu']$, and
\begin{flalign*}
 &\pE[\E_i[(1 - w'_i) \ip{x'_i  - \sampmean, v}]^2] \leq O(\eps) \cdot (  \pE V(\mu', v) - V(\sampmean, v)) + \tilde{O}(\eps^2) \cdot (\pE V(\mu', v) + V(\sampmean, v)) \enspace.
\end{flalign*}
\end{lemma}
One should think of $V$ as the variance of the distribution from which the $x_i$'s are drawn, along the direction $v \in \R^d$.
We now turn to the proof of \cref{lem:technicallem}.
\begin{proof}[Proof of \cref{lem:technicallem}]
For each $i \in [n]$, let $w'_i = w_i \cdot \mathbf{1}(x_i = y_i)$. One should think of $w_i$ as indicating that the algorithm ``thinks'' that $x_i = y_i$; the variable $w'_i$ then indicates that the algorithm correctly ``thinks'' that $x_i = y_i$.

We now notice that the constraints ${w'_i}^2 = w'_i$, $w'_i x'_i = w'_i x_i$, and $\sum_i w'_i \geq (1 - 2 \eps) n$ are all satisfied by $\pE$. Indeed, e.g., the fact that $w'_i x'_i = w'_i x_i$ is satisfied is consistent with the logic that if the algorithm thinks that $x_i = y_i$, then it chooses $x'_i = y_i = x_i$, and therefore $x'_i = x_i$.

We then have
\begin{alignat}{9}
\abs{\ip{\hat{\mu} - \sampmean, v}} &= \abs{\pE \E_i \ip{x'_i - x_i, v}} \\&= \abs{\pE\E_i w'_i \ip{x'_i - x_i, v} + \pE\E_i(1 - w'_i) \ip{x'_i  - x_i, v}} \\
&= 0 + \abs{ \pE\E_i(1 - w'_i) \ip{x'_i  - \sampmean + \sampmean - x_i, v}} \ \qq{(because $\pE$ satisfies $w'_i x'_i = w'_i x_i$)} \\
&\leq \abs{ \pE\E_i(1 - w'_i) \ip{x'_i  - \sampmean, v}}+ \abs{ \E_i\pE[1 - w'_i] \ip{x_i - \sampmean, v}} \enspace.
\end{alignat}
One should notice that in the calculation above, we split the estimation error into the term when the algorithm ``thinks'' correctly and the error term, and then we ``center'' the error term about the sample mean $\sampmean$.

We now apply the robustness assumption \eqref{eq:robustness-assumption} to the second error term. Let $a_i = \pE[w'_i]$ for each $i$. We have that $\sum_{i = 1}^n a_i \geq (1 -  2 \eps) n$ and $a_i \in [0,1]$ because $\pE[{w'_i}^2] = \pE[w'_i]$, and $a \in T$ by assumption. Hence, again by assumption, we have that
\begin{flalign*}
&\abs{ \E_i\pE[w'_i] \ip{x_i - \sampmean, v}} \leq \tilde{O}(\eps) \sqrt{V(\sampmean,v)} \qq{and} \abs{ \E_i\ip{x_i - \sampmean, v}} \leq \tilde{O}(\eps) \sqrt{V(\sampmean,v)} \enspace.
\end{flalign*}
Thus, $\abs{ \E_i\pE[1 - w'_i] \ip{x_i - \sampmean, v}} \leq \tilde{O}(\eps) \sqrt{V(\sampmean,v)}$. We note here that the robustness assumption we apply is an inequality that holds ``outside'' the pseudo-expectation $\pE$.

It thus remains to bound the first error term: $\abs{ \pE\E_i(1 - w'_i) \ip{x'_i  - \sampmean, v}}$. We do this by using constraint (4) to control its second moments.

First, by applying the Cauchy-Schwarz inequality, we have
\begin{flalign*}
&\abs{ \pE \E_i (1 - w'_i) \ip{x'_i  - \sampmean, v}} \leq \sqrt{\pE[\E_i [(1 - w'_i) \ip{x'_i  - \sampmean, v}]^2]} \enspace,
\end{flalign*}
and that
\begin{alignat}{9}
&\pE[\E_i[(1 - w'_i) \ip{x'_i  - \sampmean, v}]^2] \\
&\leq  \pE[ \E_i[(1 - w'_i)^2] \cdot \E_i [(1 - w'_i)^2 \ip{x'_i  - \sampmean, v}^2]] \ &\qq{(by Item (2) in \cref{lem:pEfacts})}\\
&=\pE[ \E_i[1 - w'_i] \cdot \E_i [(1 - w'_i) \ip{x'_i  - \sampmean, v}^2]] \ &\qq{(as $\pE$ satisfies ${w'_i}^2 = w'_i$)} \\
&\leq 2\eps \cdot \pE\E_i[(1 - w'_i) \ip{x'_i  - \sampmean, v}^2] \ &\qq{(as $\pE$ satisfies $\E_i[1 - w'_i] \leq 2 \eps$)}
\end{alignat}
Note that here we crucially need that $\pE$ is a degree-$\deghalf$ pseudo-expectation, as $\E_i[(1 - w'_i)^2] \cdot \E_i [(1 - w'_i)^2 \ip{x'_i  - \sampmean, v}^2]$ is a degree-$\deghalf$ polynomial in the SoS variables $x'_1, \dots, x'_n$ and $w_1, \dots, w_n$.

We thus need to control the second moments $\pE\E_i[(1 - w'_i) \ip{x'_i  - \sampmean, v}^2]$. Using constraint (4), we have that
\begin{alignat}{9}
&\pE\E_i[(1 - w'_i) \ip{x'_i  - \sampmean, v}^2]
\\&= \pE\E_i\ip{x'_i  - \sampmean, v}^2 -  \pE\E_i w'_i \ip{x'_i  - \sampmean, v}^2 \\
&= \pE\E_i\ip{x'_i  - \sampmean, v}^2 -  \E_i\pE w'_i \ip{x_i  - \sampmean, v}^2 \ &\qq{(as $\pE$ satisfies $w'_i x'_i = w'_i x_i$)} \\
&\leq \pE\E_i\ip{x'_i  - \mu' + \mu' - \sampmean, v}^2 - (1- \tilde{O}(\eps)) V(\sampmean, v) \ &\llap{\qq{(by \cref{eq:robustness-assumption}, setting $a_i = \pE[w'_i]$)}} \\
&=\mathrlap{\pE\E_i\ip{x'_i  - \mu'}^2 + \pE \ip{\mu' - \sampmean, v}^2 + \pE\E_i\ip{x'_i  - \mu', v}\ip{\mu' - \sampmean, v}- (1 - \tilde{O}(\eps)) V(\sampmean, v)} \\
&\leq \pE[(1 + \tilde{O}(\eps)) V(\mu', v)] + \pE\ip{\mu' - \sampmean, v}^2 + 0 - (1 - \tilde{O}(\eps)) V(\sampmean, v) &\qq{(by constraint (4))} \\
&= \mathrlap{ \pE V(\mu', v) - V(\sampmean, v) + \tilde{O}(\eps)(\pE V(\mu', v) + V(\sampmean, v)) + \pE\ip{\mu' - \sampmean, v}^2 \enspace. }
\end{alignat}
Again, we remark that \cref{eq:robustness-assumption}, used above to \emph{lower bound} the second moment, is an inequality that holds ``outside'' $\pE$.

Finally, we upper bound $\pE\ip{\mu' - \sampmean, v}^2$. We compute
\begin{alignat}{9}
&\pE\ip{\mu' - \sampmean, v}^2 
\\&= \pE[\E_i[(w'_i + (1 - w'_i))\ip{x'_i - x_i, v}]^2] \\
&=\pE[\E_i[(1 - w'_i)\ip{x'_i - x_i, v}]^2] \ &\qq{(as $\pE$ satisfies $w'_i x'_i = w'_i x_i$)}\\
&\leq \pE[ \E_i [(1 - w'_i)^2] \cdot \E_i[\ip{x'_i - x_i, v}^2]] \leq 2 \eps \pE \E_i\ip{x'_i - x_i, v}^2 \\
&= 2\eps \pE \E_i\ip{(x'_i - \mu') + (\mu' - \sampmean) + (\sampmean -  x_i), v}^2 \\
&\le6\eps \pE \E_i \ip{x'_i - \mu',v}^2 + 6\eps \E_i \ip{x_i - \sampmean,v}^2 + 6 \eps \pE\ip{\mu' - \sampmean, v}^2 &\qq{(by \cref{lem:pEfacts})} \\
&\leq 6 \eps (1 + \tilde{O}(\eps))(\pE V(\mu', v) + V(\sampmean,v)) + 6 \eps \pE \ip{\mu' - \sampmean, v}^2 \enspace,
\end{alignat}
and so it follows that $\pE \ip{\mu' - \sampmean, v}^2 \leq O(\eps) (\pE V(\mu', v) + V(\sampmean,v))$.

Putting everything together, we conclude that:
\begin{flalign*}
 \pE[\E_i[(1 - w'_i) \ip{x'_i  - \sampmean, v}]^2] \leq O(\eps) \cdot (  \pE V(\mu', v) - V(\sampmean, v)) + \tilde{O}(\eps^2) \cdot (\pE V(\mu', v) + V(\sampmean, v))
\end{flalign*}
and that 
\begin{flalign*}
\abs{\ip{\hat{\mu} - \sampmean, v}} \leq \tilde{O}(\eps)\sqrt{V(\sampmean, v)} + \sqrt{ \pE [\E_i [(1 - w'_i) \ip{x'_i  - \sampmean, v}]^2]} \enspace,
\end{flalign*}
for every $v$ in $S$.
\end{proof}

\bibliographystyle{alpha}

{\small \bibliography{sos-estimation.bbl}}

\appendix


\appendix

\section{Proof of \cref{cor:finalalg}}
\label{sec:finalalg}
We prove \cref{cor:finalalg}. Let $x_1, \dots, x_{2n}$ be drawn from $\cN(\mu, \Sigma)$, and let $y_1, \dots, y_{2n}$ be an $\eps$-corruption of $x_1, \dots, x_{2n}$. Note that $y_1, \dots, y_n$ is an $\eps$-corruption of $x_1, \dots, x_n$, and $y_{n+1}, \dots, y_{2n}$ is an $\eps$-corruption of $x_{n+1}, \dots, x_{2n}$.

First, we run the algorithm in \cref{thm:main} on $y_1, \dots, y_n$: this yields an estimate $\hat{\mu}$ satisfying Item (1) of \cref{cor:finalalg}, and an estimate $\Sigma_1$ of $\Sigma$ satisfying $\abs{v^{\top} (\Sigma_1 - \Sigma) v} \leq \tilde{O}(\eps) v^{\top} \Sigma v$ for all $v \in \R^d$.

Next, we run the algorithm in \cref{thm:frobnorm} on the transformed samples $\Sigma_1^{-1/2} y_{n+1}, \dots, \Sigma_1^{-1/2} y_{2n}$. We observe that these samples are an $\eps$-corruption of $\Sigma_1^{-1/2} x_{n+1}, \dots, \Sigma_1^{-1/2} x_{2n}$, which are drawn from $\cN(\mu, \Sigma_2)$, where $\Sigma_2 = \Sigma^{-1/2} \Sigma \Sigma^{-1/2}$. By our guarantee on $\Sigma_1$, we must have $(1 - \tilde{O}(\eps)) \Id \preceq \Sigma_2 \preceq (1 + \tilde{O}(\eps)) \Id$. Hence, the output of the algorithm in \cref{thm:frobnorm} is $\Sigma_3$ where $\Sigma_3$ satisfies $\norm{\Sigma_2^{-1/2} \Sigma_3 \Sigma_2^{-1/2} - \Id}_F \leq \tilde{O}(\eps)$.

Our final estimate for $\Sigma$ is $\hat{\Sigma} := \Sigma_1^{1/2} \Sigma_3 \Sigma_1^{1/2}$. We have that
\begin{flalign*}
&\norm{\Sigma^{-1/2} \hat{\Sigma} \Sigma^{-1/2} - \Id}_F = \norm{\Sigma^{-1/2} (\Sigma_1^{1/2} \Sigma_3 \Sigma_1^{1/2}) \Sigma^{-1/2} - \Id}_F\\
&= \norm{(\Sigma^{-1/2} \Sigma \Sigma^{-1/2})^{-1/2} \Sigma_3 (\Sigma^{-1/2} \Sigma \Sigma^{-1/2})^{-1/2} - \Id}_F = \norm{\Sigma_2^{-1/2} \Sigma_3 \Sigma_2^{-1/2} - \Id}_F \leq \tilde{O}(\eps) \enspace,
\end{flalign*}
where we use the following proposition. This finishes the proof of \cref{cor:finalalg}, as by Corollary 2.14 in~\cite{DiakonikolasKK016} we have the desired bound on the total variation distance.
\begin{proposition}
\label{prop:basischange}
Let $A,B \in \R^{d \times d}$ be symmetric, PSD matrices, with $B$ invertible. Then for any invertible $C \in \R^{d \times d}$, it holds that
\begin{equation*}
\norm{B^{-1/2} A B^{-1/2} - \Id}_F = \norm{(C B C^{\top})^{-1/2} C A C^{\top} (C B C^{\top})^{-1/2} - \Id}_F
\end{equation*}
\end{proposition}
\begin{proof}
Recall that for a symmetric matrix $M \in \R^{d \times d}$ with eigenvalues $\lambda_1, \dots, \lambda_d$, $\norm{M}_F = \sum_{i = 1}^d \lambda_i^2$. It thus suffices to show that $B^{-1/2} A B^{-1/2}$ and $(C B C^{\top})^{-1/2} C A C^{\top} (C B C^{\top})^{-1/2}$ are equivalent up to an orthogonal change of basis.

Let $D =  (C B C^{\top})^{-1/2} C B^{1/2}$. We then have that $D B^{-1/2} A B^{-1/2} D^{\top} = (C B C^{\top})^{-1/2} C A C^{\top} (C B C^{\top})^{-1/2}$, so it remains to show that $D$ is orthogonal, i.e., $D D^{\top} = D^{\top} D = \Id$.
We have
\begin{flalign*}
&D D^{\top} = (C B C^{\top})^{-1/2} C B^{1/2} B^{1/2} C^{\top} (C B C^{\top})^{-1/2} = (C B C^{\top})^{-1/2} \left((C B C^{\top})^{1/2}\right)^2 (C B C^{\top})^{-1/2} = \Id \\
&D^{\top} D = B^{1/2} C^{\top} (C B C^{\top})^{-1/2} (C B C^{\top})^{-1/2} C B^{1/2} = B^{1/2} C^{\top} (C^{-1})^{\top} B^{-1} C^{-1} C B^{1/2} = \Id \enspace,
\end{flalign*}
which finishes the proof.
\end{proof}

\section{Quantifier Elimination in Sum-of-Squares}\label{sec:elimination}
In this section we will justify why the SoS relaxations of \cref{program:canonical,program:frobnorm}, which are written as a family of infinitely many constraints, can be solved efficiently. The programs have the form
\begin{gather}\begin{aligned}
    \qq{find} &x \in \R^m 
    \\\qq{s.t.} &f_i(x) \ge 0 &&\forall i
    \\&g_j(x) = 0 &&\forall j
    \\& h(x, v) \ge 0 &&\forall v \in \R^d.
\end{aligned}\label{eq:csp2}\end{gather}
with $\poly(m)$ constraints $f_i, g_j$. As such, we need a way to express constraints of the form ``$h(x, v) \ge 0$ for all $v\in \R^d$'' within degree-$k$ SoS. This will follow from the following result:

\begin{lemma}[Quantifier elimination in SoS, e.g., Section~4.3.4 in~\cite{FlemingKP19}]\label{lem:elimination}
Suppose that there exists some $x^* \in \R^m$ such that $f_i(x^*) \ge 0$ for all $i$, $g_j(x^*) = 0$ for all $j$, and $h(x^*, \cdot)$ has a degree-$k$ SoS proof of nonnegativity. Then a degree-$k$ pseudoexpectation satisfying all constraints in~\eqref{eq:csp2} can be found by solving a semidefinite program of size $m^{O(k)}$. 
\end{lemma}
Intuitively, this is true because ``$h(x^*, \cdot)$ has a degree-$k$ SoS proof of nonnegativity'' is equivalent to a particular moment matrix of size $m^{O(k)}$ being PSD, which can be expressed within SoS.
We will use this result in two forms. 
\begin{lemma}[Lemma~4.27 in~\cite{FlemingKP19}]
If $h$ is a quadratic form, then $h$ has a degree-$2$ SoS proof of nonnegativity if and only if $h(v) \ge 0$ for every $v$.
\end{lemma}

In \cref{program:canonical,program:frobnorm}, the constraint (4) is indeed a quadratic form in $v$ (or $P$), so we are done. For constraint (5) in \cref{program:canonical}, we require the certifiable hypercontractivity of Gaussians (\cref{lem:certifiable-hypercontractivity}), which we restate here:

\restatelemma{lem:certifiable-hypercontractivity}

By rearranging, the constraint (5) in \cref{program:canonical} is exactly the condition that $h(x', v) \ge 0$ for all $v$ in the above lemma, so the lemma states that a degree-$4$ SoS proof exists with high probability for the {\em true} samples $x'_i = x_i$. Thus, the conditions of \cref{lem:elimination} are satisfied, and we are done.

\section{Deferred Proofs from \cref{sec:gaussianproperties}}
\label{sec:concproof}
\subsection{Proof of \cref{lem:resiliencecov}}

\begin{proof}
The first statement is
    \begin{align}
        \bignorm{\E_i a_i [ x_i x_i^\top - \Id]}_F \le \tilde O(\eps)
    \end{align}
    which is Corollary~4.8 in~\cite{DiakonikolasKK016}.
    The second statement is
    \begin{align}
        \bigabs{\E_i a_i [ \ip{x_i x_i^\top - \Id, P}^2 - 2\norm{P}_F^2 ] } \leq \tilde{O}(\eps) \norm{P}_F^2.
    \end{align}
    By convexity, we may assume that $a_i \in \{0, 1\}$ for all $i$. Let $S$ be the set of indices $i$ for which $a_i > 0$. 
    We have:
    \begin{align}
        \bigabs{\E_i a_i [ \ip{x_i x_i^\top - \Id, P}^2 - 2\norm{P}_F^2 ] } &\le \bigabs{ \E_{i \sim [n] } \ip{x_i x_i^\top - \Id, P}^2 - 2\norm{P}_F^2 } + \eps \bigabs{ \E_{i \sim [n] \setminus S} \ip{x_i x_i^\top - \Id, P}^2}
    \end{align}
    We bound the two terms separately. Condition on the ``good event'' in Lemma~5.17 of~\cite{DiakonikolasKK016}. Then,
    \begin{align}
        \bigabs{\E_{i \sim [n] } \ip{x_i x_i^\top - \Id, P}^2 - 2 \norm{P}_F^2} \le O(\eps)\norm{P}_F^2
    \end{align}
    follows from Item~3 of Definition~5.15 in~\cite{DiakonikolasKK016} with $p(x) = \ip{xx^T - \Id, P} / (\sqrt{2} \norm{P}_F)$. The fact that 
    \begin{align}
        \eps \bigabs{ \E_{i \sim [n] \setminus S} \ip{x_i x_i^\top - \Id, P}^2} \le \tilde O(\eps)
    \end{align}
    follows from Lemma~5.21 of~\cite{DiakonikolasKK016} with the same choice of $p$. Combining these bounds completes the proof. 
\end{proof}

\subsection{Proof of \cref{lem:annoying-statements}}
The statements in \cref{lem:annoying-statements} are similar to those in \cref{lem:resiliencemean} and \cref{lem:resiliencecov}, so it should be reasonable to believe that they should hold. The proofs are tedious but ultimately mostly brute force. 

All the statements are invariant to linear transformations, so assume WLOG that $\mu = 0$ and $\Sigma = \Id$. Condition on the conclusions of \cref{lem:resiliencemean,lem:resiliencecov}, which hold with high probability for the chosen $n$. Let $z_i = x_i - \mu_0$ for notational simlpicity.

In the proof, instead of the stated conditions on $a$, we will use instead normalized vectors, namely, we will assume that $\E_{ij} a_{ij} = \E_i a_i = 1$,  $\E_j a_{ij} = a_i \le 1 + O(\eps)$, and $a_{ij} \le a_i(1 + O(\eps))$. Since this amounts to nothing but scaling the coefficients by $1 + O(\eps)$, the conclusions of \cref{lem:resiliencemean,lem:resiliencecov} hold verbatim.
\begin{enumerate}
\item $\displaystyle \abs{\ip{\mu_0, v}} \leq \tilde{O}(\eps) \norm{v}_2$
\end{enumerate}
\begin{proof}Set $a_i = 1$ for all $i$ in \cref{lem:resiliencemean}. \end{proof}

\begin{enumerate}[resume]
\item$\displaystyle \bigabs{\E_i a_i \ip{x_i - \mu_0, v}} \leq \tilde{O}(\eps) \cdot \sqrt{v^{\top} \Sigma_0 v}$
\end{enumerate}
\begin{proof}
By (1) above and \cref{lem:resiliencemean}, we have
\begin{align}
    \bigabs{\E_i a_i \ip{x_i - \mu_0, v}} \le \bigabs{\E_i a_i \ip{x_i, v}} + \bigabs{\E_i a_i \ip{\mu_0, v}} \le \tilde O(\eps) \norm{v}_2
\end{align}
But $\norm{v}_2 = (1 \pm \tilde O(\eps)) \sqrt{v^\top \Sigma_0 v}$ by (4) (with $P = vv^\top$), so we are done. \end{proof}

The following intermediate results will be useful in the remaining proofs.
\begin{lemma}\label{lem:cross-terms}
$\norm{P \mu_0}_2 \le \tilde O(\eps) \norm{P}_F$, $\abs{\E_i a_i \ip{x_i \mu_0^\top, P}} \le \tilde O(\eps^2) \norm{P}_F$, and $\abs{\mu_0^\top P \mu_0} \le \tilde O(\eps^2) \norm{P}_F$.
\end{lemma}
\begin{proof}
The first inequality is $\norm{P \mu_0}_2 \le \norm{P}_2 \norm{\mu_0}_2 \le \tilde O(\eps) \norm{P}_F$  by \cref{lem:resiliencemean} and $\norm{P}_2 \le \norm{P}_F$.

The second is $\abs{\E_i a_i \ip{x_i \mu_0^\top, P}} \le \tilde O(\eps) \norm{P \mu_0}_2 \le \tilde O(\eps^2) \norm{P}_F$, by  \cref{lem:resiliencemean} the first inequality.

The third is the second when $a_i = 1$ for all $i$.
\end{proof}

\begin{enumerate}[resume]
\item$\displaystyle \bigabs{\E_i a_i \ip{z_i z_i^\top - \Id, P}} \leq \tilde{O}(\eps) \norm{P}_F$
\end{enumerate}
(The statement in the lemma follows from setting $P = vv^\top$ and applying (4) below, but we will need this more generic statement later, so this is the one we prove.)
\begin{proof}
We have
\begin{align}
    \E_i a_i  \ip{(x_i - \mu_0)(x_i - \mu_0)^\top, P} &= \E_i a_i \ip{x_i x_i^\top, P} + \ip{\mu_0 \mu_0^\top, P} - 2 \E_i a_i \ip{x_i \mu_0^\top, P}
\end{align}
The first term is $\E_i a_i \ip{x_i x_i^\top, P} = \ip{\Id, P} \pm \tilde O(\eps) \norm{P}_F$ by \cref{lem:resiliencecov}, and the other terms are $\pm \tilde O(\eps) \norm{P}_F$ by \cref{lem:cross-terms}.  
\end{proof}

\begin{enumerate}[resume]
\item$\displaystyle \bigabs{\ip{\Sigma_0 - \Id, P}} \leq \tilde{O}(\eps) \norm{P}_F$
\end{enumerate}
\begin{proof}
Set $a_i = 1$ for all $i$ in (3).
\end{proof}

\begin{enumerate}[resume]
\item$\displaystyle \bigabs{\E_{i} \ip{z_i z_i^\top - \Sigma_0, P}^2 - 2\norm{P}_F^2} \leq \tilde{O}(\eps) \cdot \norm{P }_F^2$
\end{enumerate}
\begin{proof}
We have:
\begin{align}
    \E_{i} a_{i} \ip{z_i z_i^\top - \Sigma_0, P}^2 =  \E_{i} a_{i} \ip{z_i z_i^\top  - \Id, P}^2 + \ip{\Sigma_0 - \Id, P}^2 - 2  \E_{i} a_{i} \ip{z_i z_i^\top  - \Id, P} \ip{\Sigma_0 - \Id, P}
\end{align}
For the second term, we have $\ip{\Sigma_0 - \Id, P}^2 \le \tilde O(\eps^2) \norm{P}_F^2$ by (4). For the third term, we have 
\begin{align}
    \bigabs{\E_{i} a_{i} \ip{z_i z_i^\top - \Id, P} \ip{\Sigma_0 - \Id, P}}
    &\le \tilde O(\eps)  \norm{P}_F  \bigabs{\E_{i} a_{i} \ip{z_i z_i^\top - \Id, P}}
    \le \tilde O(\eps^2) \norm{P}_F^2
\end{align}
by (4) and then (3). That only leaves the first term. We have:
\begin{align}
    \E_i a_i \ip{(x_i - \mu_0)(x_i - \mu_0)^\top - \Id, P}^2 &= \E_i a_i \bigiprod{ (x_i x_i^\top - \Id) + \mu_0 \mu_0^\top - 2 x_i \mu_0^\top, P }^2
    \\&= \E_i a_i \ip { x_i x_i^\top - \Id, P}^2 + \ip{\mu_0 \mu_0^\top, P}^2 + 4 \E_i a_i \ip{x_i \mu_0^\top, P}^2
    \\&\quad + \ip{\mu_0 \mu_0^\top, P} \E_i a_i \ip { x_i x_i^\top - \Id, P} - 2\ip{\mu_0 \mu_0^\top, P}\E_i a_i \ip{x_i \mu_0^\top, P}
    \\&\quad  - 2 \E_i a_i \ip { x_i x_i^\top - \Id, P} \ip{x_i \mu_0^\top, P}
\end{align}
We analyze each term separately. The first term is $(2 \pm \tilde O(\eps)) \norm{P}_F^2$ by \cref{lem:resiliencecov}, so it suffices to show that all remaining terms are small. The second term is $\tilde O(\eps^4) \norm{P}_F^2$ by \cref{lem:cross-terms}. The third term is $\abs{\E_i a_i \ip{x_i, P \mu_0}^2} = (1 \pm \tilde O(\eps)) \norm{P \mu_0}_2^2 \le \tilde O(\eps^2) \norm{P}_F^2$ by \cref{lem:resiliencemean,lem:cross-terms}. The fourth term is $\tilde O(\eps^3) \norm{P}_F^2$ by \cref{lem:cross-terms,lem:resiliencecov}. The fifth term is $\tilde O(\eps^4) \norm{P}_F^2 $ by \cref{lem:cross-terms}. For the final term, we have
\begin{align}
    \bigabs{ \E_i a_i \ip { x_i x_i^\top - \Id, P} \ip{x_i \mu_0^\top, P}} \le \sqrt{\qty\big(\E_i a_i \ip{ x_i x_i^\top - \Id, P}^2) \qty\big(\E_i a_i \ip{x_i \mu_0^\top, P}^2)}
\end{align}
by Cauchy-Schwarz. The first term is $\E_i a_i \ip{ x_i x_i^\top - \Id, P}^2 = (2 \pm \tilde O(\eps))\norm{P}_F^2$, and the second term is $\tilde O(\eps^2) \norm{P}_F^2$ as argued above. Combining these yields $\bigabs{ \E_i a_i \ip { x_i x_i^\top - \Id, P} \ip{x_i \mu_0^\top, P}} \le \tilde O(\eps) \norm{P}_F^2$, as needed.
\end{proof}

The proofs of the remaining two bounds will make repeated use of the following generic technique, roughly speaking. Suppose that resilience (\cref{lem:resiliencemean,lem:resiliencecov}) gives us $\abs{\E_i a_i z_i} = \tilde O(\eps)B$, and we want to argue that $\abs{\E_{ij} a_{ij} z_i z_j} = \tilde O(\eps)B^2$. This is not immediate, because $a_{ij} \ne a_i a_j$ in general. Instead, we write  $\E_{ij} a_{ij} z_i z_j = \E_i z_i \E_j a_{ij} z_j$, and apply resilience to the inner expectation (noting that the vector whose $j$th entry is $a_{ij} / a_i$ is, by construction, a valid resilience vector)  to find $\abs{\E_{j} a_{ij} z_j} \le a_i \tilde O(\eps)B$ for each $i$, so that $\abs{\E_{ij} a_{ij} z_i z_j} \le B \abs{\E_i z_i \tilde O(\eps)}$. In this expression, the $\tilde O(\eps)$ may depend on $i$. Let $b_i$ be the term hidden by the $\tilde O(\eps)$ for each $i$, and let $a'_i := 1 - b_i + \E b_i$. Note that $\E_i a'_i = 1$ and $a'_i = 1 \pm \tilde O(\eps)$ for all $i$, so $a'$ is a valid input to the resilience condition. Thus, we have
\begin{align}
    \abs{\E_i b_i z_i} \le ( 1 + \E_i b_i ) \abs{\E_i z_i} + \abs{\E_i a_i z_i}
\end{align}
Now applying resilience to each of the two terms separately gives $\abs{\E_i b_i z_i} \le \tilde O(\eps)B$, so $\abs{\E_{ij} a_{ij} z_i z_j} \le \tilde O(\eps) B^2$, as desired.

The following intermediate result will also be useful:
\begin{lemma}
\label{lem:paired-to-unpaired}
$\displaystyle \E_{ij} a_{ij} \ip{ X_{ij}, P }= \E_{i} a_{i} \ip{x_i x_i^\top, P} \pm \tilde O(\eps) \norm{P}_F$
\end{lemma}
\begin{proof}
We have
\begin{align}
    \E_{ij} a_{ij} \ip{ X_{ij}, P }&= \frac12 \E_{ij} a_{ij} \ip{(x_i - x_j)(x_i - x_j)^\top, P}
    \\&= \E_{i} a_{i} \ip{x_i x_i^\top, P} - \E_{ij} a_{ij} \ip{x_i x_j^\top, P}
\end{align}
where we use the symmetry of $P$ and the $a_{ij}$s. It thus remains to bound the last term.  Write $P = \sum_k \lambda_k v_k v_k^\top$ for orthonormal vectors $v_k$. Note that $\norm{\sum_k \lambda_k v_k}_2 = \sqrt{\sum_k \lambda_k^2} = \norm{P}_F$ by Pythagorean theorem. Then:
\begin{align}
    \E_{ij} a_{ij} \ip{x_i x_j^\top, P} &=  \E_{i}\sum_k \lambda_k \ip{v_k, x_i} \E_j a_{ij} \ip{v_k, x_j}
    = \E_{i} a_i  (\pm \tilde O(\eps))   \Bigiprod{\sum_k \lambda_k v_k, x_i}
    = \tilde O(\eps) \norm{P}_F
\end{align}
by applying \cref{lem:resiliencemean} twice using the generic technique.

\end{proof}
We now prove the last two results in \cref{lem:annoying-statements}.

\begin{enumerate}[resume]
\item$\displaystyle \bigabs{\E_{ij}a_{ij} \ip{X_{ij} - \Sigma_0, P}} \leq \tilde{O}(\eps) \cdot \norm{P}_F$
\end{enumerate}
\begin{proof}
By \cref{lem:paired-to-unpaired}, we have:
\begin{align}
    \E_{ij} a_{ij} \ip{ X_{ij} - \Sigma_0, P }&= \E_{i} a_{i} \ip{x_i x_i^\top - \Sigma_0, P} \pm \tilde O(\eps^2) \norm{P}_F
\end{align}
It thus only remains to bound the first term. We have
\begin{align}
    \bigabs{\E_i a_i \ip{x_i x_i^\top - \Sigma_0, P}} \le \bigabs{\E_i a_i \ip{x_i x_i^\top - \Id, P}} + \bigabs{\E_i a_i \ip{\Sigma_0 - \Id, P}} \le \tilde O(\eps) \norm{P}_F
\end{align}
by applying \cref{lem:resiliencecov} and (4).  \end{proof}

\begin{enumerate}[resume]
\item$\displaystyle \bigabs{\E_{ij} a_{ij} \ip{X_{ij} - \Sigma_0, P}^2 - 2\norm{P}_F^2} \leq \tilde{O}(\eps) \cdot \norm{P}_F^2$
\end{enumerate}
\begin{proof}
We follow the same structure as the proof of (5) above. We have:
\begin{align}
    \E_{ij} a_{ij} \ip{X_{ij} - \Sigma_0, P}^2 =  \E_{ij} a_{ij} \ip{X_{ij} - \Id, P}^2 + \ip{\Sigma_0 - \Id, P}^2 - 2 \E_{ij} a_{ij} \ip{X_{ij} - \Id, P} \ip{\Sigma_0 - \Id, P}
\end{align}
For the second term, we have $\ip{\Sigma_0 - \Id, P}^2 \le \tilde O(\eps^2) \norm{P}_F^2$ by (4). For the third term, we have 
\begin{align}
    \bigabs{\E_{ij} a_{ij} \ip{X_{ij} - \Id, P} \ip{\Sigma_0 - \Id, P}}
    &\le \tilde O(\eps)  \norm{P}_F  \bigabs{\E_{ij} a_{ij} \ip{X_{ij} - \Id, P}}
    \\&=  \tilde O(\eps) \norm{P}_F \bigabs{\E_{i} a_{i} \ip{x_i x_i^\top - \Id, P} \pm \tilde O(\eps) \norm{P}_F}
    \\&\le \tilde O(\eps^2) \norm{P}_F^2
\end{align}
by \cref{lem:paired-to-unpaired,lem:resiliencecov}. 
That only leaves the first term. We have
\begin{align}
    \E_{ij} a_{ij} \ip{X_{ij} - \Id, P}^2 &= \E_{ij} a_{ij} \Bigiprod{\frac12 (x_i x_i^\top - \Id) + \frac12 (x_j x_j^\top - \Id) - x_i x_j^\top , P}^2
    \\&= \frac12 \E_{i} a_{i} \ip{x_i x_i^\top - \Id, P}^2 + \E_{ij} a_{ij} \ip{x_i x_j^\top, P}^2 + \frac12 \E_{ij} a_{ij} \ip{x_i x_i^\top - \Id, P} \ip{x_j x_j^\top - \Id, P} 
    \\&\quad - 2 \E_{ij} a_{ij} \ip{x_i x_i^\top - \Id, P} \ip{x_i x_j^\top, P}
\end{align}
The first term is $(1 \pm \tilde O(\eps)) \norm{P}_F^2$ by \cref{lem:resiliencecov}. For the second term, we have
\begin{align}
    \E_{ij} a_{ij} \ip{x_i x_j^\top, P}^2 &= \E_i \E_j a_{ij} x_i^\top P  x_j x_j^\top P x_i
    \\&= \E_i  \E_j a_{ij} \ip{x_j x_j^\top, P x_i x_i^\top P}
    \\&= \E_i a_i \qty\big[ \ip{\Id, P x_i x_i^\top P} \pm \tilde O(\eps) \norm{P x_i x_i^\top P}_F ]
    \\&=  \E_i a_i (1 \pm \tilde O(\eps)) \ip{x_i x_i^\top, P^2} 
    \\&= \ip{\Id, P^2} \pm \tilde O(\eps) \norm{P^2}_F
    \\&= (1 \pm \tilde O(\eps)) \norm{P}_F^2
\end{align}
where we use \cref{lem:resiliencecov} twice (the second time exploiting the fact that $a_i (1 \pm \tilde O(\eps))$ is still a valid resilience vector), and the last line uses $\ip{\Id, P^2} = \norm{P}_F^2$ and $\norm{P^2}_F \le \norm{P}_F^2$.

Thus, it only remains to show that the other two terms are small. For the third term, we have
\begin{align}
     \bigabs{\E_{ij} a_{ij} \ip{x_i x_i^\top - \Id, P} \ip{x_j x_j^\top - \Id, P}} &= \bigabs{\E_i \ip{x_i x_i^\top - \Id, P} \E_j a_{ij} \ip{x_j x_j^\top - \Id, P}}
     \\&\le  \norm{P}_F \bigabs{ \E_i a_i  (\pm \tilde O(\eps))  \ip{x_i x_i^\top - \Id, P} }
     \\&\le \tilde O(\eps) \norm{P}_F^2.
\end{align}
by the generic technique. For the final term, we have 
\begin{align}
    \bigabs{\E_{ij} a_{ij} \ip{x_i x_i^\top - \Id, P} \ip{x_i x_j^\top, P}} &= \bigabs{\E_i a_{ij} \ip{x_i x_i^\top - \Id, P} \qty\big(\E_j a_{ij} \ip{x_j, P x_i})}
    \\&= \bigabs{\E_i a_{i} (\pm \tilde O(\eps))    \ip{x_i x_i^\top - \Id, P}  \norm{P x_i}_2}
    \\&\le \sqrt{\qty\big(\E_i a_i \ip{x_i x_i^\top - \Id, P}^2) \qty\big( \E_i a_i \tilde O(\eps^2) \norm{P x_i}_2^2) }
\end{align}
The first term is $\E_i a_i \ip{x_i x_i^\top - \Id, P}^2 = (2 \pm \tilde O(\eps))\norm{P}_F^2$. For the second term, we have
\begin{align}
    \bigabs{\E_i a_i \tilde O(\eps^2) \norm{P x_i}_2^2} &= \tilde O(\eps^2) \bigabs{ \E_i  \ip{x_i x_i^\top, P^2} } 
    \\&\le \tilde O(\eps^2) \bigabs{\E_i a_i  \ip{x_i x_i^\top - \Id , P^2} } + \tilde O(\eps^2) \ip{\Id, P^2}
    \\&\le \tilde O(\eps^2) \norm{P}_F^2
\end{align}
where the last line applies \cref{lem:resiliencecov} to the first term (noting again that $\ip{\Id, P^2} = \norm{P}_F^2$), and the inequality $\norm{P^2}_F \le \norm{P}^2_F$ to the second. Combining these yields $\bigabs{\E_{ij} a_{ij} \ip{x_i x_i^\top - \Id, P} \ip{x_i x_j^\top, P}} \le \tilde O(\eps) \norm{P}_F^2$, which is what we needed. \end{proof}

\end{document}